\documentclass[11pt]{article}
\newcommand{\ignore}[1]{}
\newcommand\Amat{A}
\newcommand\Amatcomp{B}
\usepackage[normalem]{ulem}

\usepackage{amsfonts,amsmath,amsthm,times,fullpage,amssymb,epsfig,graphicx,enumerate,bm}

\usepackage{hyperref}
\hypersetup{colorlinks=false}

\usepackage[dvipsnames]{xcolor}

\usepackage{enumitem}
\usepackage{color}
\usepackage{dsfont}
\usepackage{wrapfig}
\usepackage{relsize}
\usepackage{url} 

\usepackage[left=1.0in,top=1.0in,right=1.0in,bottom=1.0in,nohead,textheight=10in,footskip=0.3in]{geometry}

\usepackage[noend]{algpseudocode} 
\usepackage[nothing]{algorithm}  
\usepackage{caption}
 
 \newtheorem*{aLLL}{Algorithmic LLL}
\newtheorem{theorem}{Theorem}[section]
\newtheorem{lemma}[theorem]{Lemma}
\newtheorem{corollary}[theorem]{Corollary}

\newtheorem{definition}[theorem]{Definition} 
\newtheorem{remark}{Remark}[section]
\newtheorem{proposition}{Proposition}[section]

\newtheorem*{gLLL}{General LLL}

\newcommand\ex{{\mathbb{E}}}

\newcommand{\param}{(1 + \delta)} 
 
\newcommand{\pot}{\Psi}

\newcommand{\beq}{\begin{equation}}
\newcommand{\eeq}{\end{equation}}

\newcommand{\primary}{primary}

\begin{document}

\title{
%A New Perspective on Stochastic Local Search and\\ 
%the Lov\'{a}sz Local Lemma
Beyond the Lov\'{a}sz Local Lemma: Point to Set Correlations and Their Algorithmic Applications
\thanks{A preliminary version of this paper appeared in the {\it 60th Annual IEEE Symposium 
on Foundations of Computer Science}.}
}

\author{
Dimitris Achlioptas
\thanks{
Research supported by NSF grant CCF-1514434} 
\\ University of California Santa Cruz
\\
{\small {optas@cs.ucsc.edu}}
\\
\and 
Fotis Iliopoulos
\thanks{Research supported by NSF grant CCF-1514434 and the Onassis Foundation} \\ 
University of California Berkeley \\
{\small {fotis.iliopoulos@berkeley.edu}}
 \\
 \and 
Alistair Sinclair
\thanks{Research supported by NSF grants CCF-1514434 and CCF-1815328} \\ 
University of California Berkeley \\
{\small {sinclair@berkeley.edu}}
}

\date{\empty}

\maketitle

\begin{abstract}
Following the groundbreaking algorithm of Moser and Tardos for the Lov\'{a}sz Local Lemma (LLL), there has been a plethora of results analyzing local search algorithms for various constraint satisfaction problems. The algorithms considered fall into two broad categories: \emph{resampling} algorithms, analyzed via different algorithmic LLL conditions; and \emph{backtracking} algorithms, analyzed via  entropy compression arguments.  This paper introduces a new convergence condition that seamlessly handles resampling, backtracking, and \emph{hybrid\/} algorithms, i.e., algorithms that perform both resampling and backtracking steps. Unlike previous work on the LLL, our condition 
replaces the notion of a dependency or causality graph by quantifying \emph{point-to-set correlations\/}
between bad events.
As a result, our condition simultaneously: (i)~captures the most general algorithmic LLL condition known as a special case; (ii)~significantly simplifies the analysis of entropy compression applications; (iii)~relates backtracking algorithms, which are conceptually very different from resampling algorithms, to the LLL; and most importantly (iv)~allows for the analysis of hybrid algorithms, which were outside the scope of previous techniques. We give several applications of our condition, including a new hybrid vertex coloring algorithm that extends the recent breakthrough result of Molloy for coloring triangle-free graphs to arbitrary graphs.

From a technical perspective, our  main new insight is that LLL-inspired
convergence arguments for local search algorithms can be viewed as methods for bounding the {\it spectral radius\/} of an associated  transition matrix.

\end{abstract}

\thispagestyle{empty}

\newpage
\setcounter{page}{1}\maketitle

\newpage

\section{Introduction}

Numerous problems in computer science and combinatorics can be formulated as searching for objects that lack certain bad properties, or ``flaws". For example,  constraint satisfaction problems such as satisfiability and graph coloring can be seen as searching for objects (truth assignments, colorings) that are {\it flawless}, in the sense that they do not violate any constraint.

The \emph{Lov\'{a}sz Local Lemma (LLL)}~\cite{LLL} is a powerful tool for proving the \emph{existence} of  flawless objects that has had far-reaching consequences in computer science and combinatorics~\cite{ASbook, mike_book}. Roughly speaking, it asserts that, given a collection of a bad events  in a probability space, if all of them are individually not too likely, and independent of most other bad events, then the probability that none of them occurs is strictly positive; hence a flawless object (i.e., an elementary event that does not belong in any bad event)  exists. For example, the LLL implies that every $k$-CNF formula in which each clause shares variables with fewer than $2^k/ \mathrm{e}$ other clauses is satisfiable. Remarkably, this is tight~\cite{2ke}.

In seminal work, Moser~\cite{moser} and Moser-Tardos~\cite{MT} showed that a simple local search algorithm can be used to make the LLL constructive for product probability spaces, i.e., spaces where elementary events correspond to sequences of independent coin-flips.  For example, the Moser-Tardos algorithm for satisfiability starts at a uniformly random truth assignment and, as long as violated clauses exist, selects any such clause  and resamples the values of all its variables uniformly at random. Following this work, a large amount of effort has been devoted to making different variants  of the LLL constructive~\cite{determ,CLLL,szege_meet,PegdenIndepen}, and to analyzing sophisticated \emph{resampling} algorithms that extend the Moser-Tardos techniques to non-product probability spaces~\cite{AIJACM,  AIK, StochControl,PermHarris,HV, molloy2017list,IS}. Indeed, intimate connections have been established between resampling algorithms and the so-called  ``lopsided" versions of the  LLL: see~\cite{AIK,HV,IS} for more details.

In earlier groundbreaking work aimed at making the LLL algorithmic for $k$-SAT, Moser~\cite{moser} introduced the \emph{entropy compression} method. This method has since been used to analyze \emph{backtracking} algorithms, e.g.,~\cite{dujmovic2016nonrepetitive,acyclic,gkagol2016pathwidth,EntropyGraphColorings,grytczuk2013new,PatternAvoidance}.
%for non-repetitive sequences~\cite{grytczuk2013new,dujmovic2016nonrepetitive}, acyclic edge coloring~\cite{acyclic}, non-repetitive list-coloring~\cite{gkagol2016pathwidth}, the Thue choice number~\cite{EntropyGraphColorings}, and pattern avoidance~\cite{PatternAvoidance}. 
These natural and potentially powerful local search algorithms have a very different flavor from resampling algorithms: instead of maintaining a complete value assignment (object) and repeatedly modifying it until it becomes satisfying, they operate on \emph{partial non-violating} assignments, starting with the empty assignment, and try to extend to a complete, satisfying assignment.  To do this, at each step they assign a (random) value to a currently unassigned variable; if this leads to the violation of one or more constraints, they backtrack to a partial non-violating assignment by unassigning some set of variables (typically including the last assigned variable).  

While there have been efforts to treat certain classes of backtracking algorithms 
systematically~\cite{acyclic,EntropyGraphColorings}, the analysis of such
algorithms in general still requires {\it ad hoc\/} technical machinery.
Moreover, there is no known connection between backtracking algorithms and any known LLL condition, either existential or algorithmic. The main reason for this is that backtracking steps induce non-trivial correlations among bad events, which typically result in very dense dependency graphs that are not amenable to currently known LLL conditions.

The main contribution of this paper is to introduce a new technique for analyzing \emph{hybrid\/} algorithms, i.e., algorithms that (potentially) use both resampling and backtracking steps. Such algorithms combine the advantages of both approaches by using resampling to explore the state space, while detecting and backing away from unfavorable regions using backtracking steps. To analyze these algorithms, we prove a new algorithmic LLL condition which, unlike previous versions,
replaces the notion of a dependency or causality graph by quantifying \emph{point-to-set correlations\/}
between bad events.
%avoids a hard notion of dependence, i.e., a dependency or causality graph, by quantifying \emph{point-to-set correlations}, so  that interactions can be arbitrarily dense as long as they are sufficiently weak. 
Notably, our new condition captures the most general algorithmic LLL condition known so far and, moreover, unifies the analysis of all entropy compression applications, connecting backtracking algorithms to the LLL in the same fashion that existing analyses connect resampling algorithms to the LLL. 

Our main  technical insight is that essentially all existing LLL-inspired
convergence arguments for local search algorithms can be viewed as methods for bounding the {\it spectral radius\/} of an associated transition matrix by  decomposing it into a sum of sparse matrices, one for each flaw. Our new condition hinges on a more refined decomposition that allows us to take advantage of higher-dimensional information.

\paragraph{A new coloring algorithm.}
Our main application is a new vertex coloring algorithm inspired by the recent breakthrough result of Molloy~\cite{molloy2017list}, who proved that any triangle-free graph of maximum degree~$\Delta$ can be colored using $(1+o(1)) \Delta / \ln \Delta $ colors; this improves the celebrated result of Johannson~\cite{JO} while, at the same time, dramatically simplifying its analysis. We generalize Molloy's result by establishing a bound on the chromatic number of \emph{arbitrary} graphs, as a function of the maximum number of triangles in the neighborhood of a vertex, and giving an algorithm that produces such a 
coloring. % \begin{theorem}[Informal Statement]\label{color_inform}
% Let $G$ be a graph with maximum degree $\Delta$ in which the neighbors of every vertex span at most $\Delta^2/f$ edges. Then, 
% \begin{align*}
% \chi(G) \le (2 + o(1) )  \Delta / \ln \sqrt{f}  \enspace.
% \end{align*}
% Further, such a vertex coloring can be found in polynomial time with high probability. Moreover, if  \\$f \in [ \Delta^{ \frac{2 + 2 \epsilon }{ 1 + 2 \epsilon } (\ln \Delta)^2}, \Delta^2 +1    ] $, then $(2+ o(1) )$ can be replaced by $(1+ o(1) )$.
% \end{theorem}

 \begin{theorem}[Informal Statement]\label{color_inform}
 Let $G$ be any graph with maximum degree $\Delta$ in which the neighbors of every vertex span at   most  $T \geq 0$ edges between them. For every $\epsilon > 0$, if $\Delta \ge \Delta_{\epsilon}$ and $T \lesssim \Delta^{2 \epsilon} $ then
 \begin{align}\label{color_inform_bound}
 \chi(G) \le (1 + \epsilon )  \frac{\Delta}{\ln \Delta - \frac{1}{2}\ln(T+1)} ,
 \end{align}
and such a vertex coloring can be found efficiently. (Here $\lesssim$ hides logarithmic factors.) Moreover,  the theorem holds for any $T \ge 0 $ if the leading constant $(1+ \epsilon)$ is replaced by $(2+\epsilon)$.
 \end{theorem}

\par\noindent 
Importantly, as explained in Section~\ref{subsec:results:coloring},
the bound~\eqref{color_inform_bound} \emph{matches} the algorithmic barrier for random graphs~\cite{mitsaras_barriers}. This implies that any improvement on the guarantee of our algorithm for $T \lesssim \Delta^{2\epsilon}$ would amount to an unexpected breakthrough in random graph theory. (Random graphs are  only informative in the regime $ T  \lesssim \Delta^{2 \epsilon}$.) For arbitrary graphs our bound is within a factor of 4 of the chromatic number, improving upon a classical result of Alon, Krivelevich and Sudakov~\cite{alon1999coloring} who showed~\eqref{color_inform_bound} with an unspecified (large) leading constant.

At the heart of Theorem~\ref{color_inform} is a hybrid local search algorithm which we analyze using the techniques introduced in this paper.  Molloy's result (and resampling algorithms based on it) breaks
down immediately in the presence of triangles; the key to our algorithm is to allow backtracking steps in order to avoid undesirable portions of the search space. We discuss our coloring result and its optimality further in Section~\ref{subsec:results:coloring}.

%\newpage

\paragraph{Applications to backtracking algorithms.}

Besides our coloring application, we give three representative applications of our techniques applied to
pure backtracking algorithms. 

Recently, Bissacot and Doin~\cite{bissacot2017entropy} showed that backtracking algorithms can make LLL applications in the variable setting constructive, using the entropy compression method. However, their result applies only to the uniform measure and their algorithms are relatively complicated. 
Our new algorithmic condition makes applications of the LLL in the variable setting~\cite{MT}, with any measure, constructive via a single, simple backtracking algorithm, i.e., an algorithm of very different flavor from the Moser-Tardos resampling algorithm. For example, in the case of $k$-SAT the algorithm takes the following form:

\begin{algorithm}[H]
\caption*{{\bf Randomized DPLL with single-clause backtracking }}
\begin{algorithmic}%[1]
\While {unassigned variables exist } 
\State Assign to the lowest indexed unassigned variable $x$ a value $v \in \{0,1 \}$ with probability $p_x^v$ 
\If {one or more clauses become violated}
\State Unassign all $k$ variables in the lowest indexed violated clause 
\EndIf
\EndWhile
\end{algorithmic}
\end{algorithm}
%While the analysis of this algorithm is completely outside the scope of current techniques, it becomes trivial using our condition.

We then show that applying our condition to the algorithm of Esperet and Parreau~\cite{acyclic} for \emph{acyclic edge coloring} recovers, in a simple, black-box fashion, the same bound of $4 \Delta$  as their highly non-trivial, problem-specific analysis via entropy compression, while guaranteeing an improved running time bound.

Finally, we make constructive in an effortless manner an existential result of Bernshteyn~\cite{bernshteyn2016new} showing improved bounds for the acyclic chromatic index of graphs that do not contain an arbitrary bipartite graph~$H$. 

We present our results on backtracking algorithms in Section~\ref{SoftcoreLLL}.

\subsection{A new algorithmic LLL condition}\label{sec:informal}

%The matrix-norm perspective we introduce in this paper allows us not only to cast the probabilistic method aspect of the algorithmic LLL as a change of basis, and the overall approach as a dual potential function argument, but, more importantly, to significantly expand and refine the analysis, so that it can avoid a hard notion of dependence, i.e., a dependency or causality graph. This is because, unlike past works, our condition quantifies \emph{point-to-set correlations}, so that interactions can be arbitrarily dense as long as they are sufficiently weak. Before stating our result we need to fix some notation.

To state our new algorithmic LLL condition, we need some standard terminology. Let $\Omega$ be a finite set and let $F = \{f_1, f_2, \ldots, f_m \}$ be a collection of subsets of $\Omega$, each of which will be referred to as a ``flaw."  Let $\bigcup_{i \in [m]} f_i = \Omega^*$.    For example, for a given CNF formula on $n$ variables with clauses $c_1,\ldots,c_m$, we take $\Omega=\{0,1\}^n$ to be the set of all possible variable assignments, and $f_i$ the set of assignments that fail to satisfy clause~$c_i$. Our goal is to find an assignment in $\Omega\setminus\Omega^*$, i.e., a satisfying (``flawless") assignment. Since we will be interested in algorithms which traverse the set $\Omega$ we will also refer to its elements as {\it states}.

%For a state $\sigma$, we denote by $U(\sigma) = \{ j \in [m] : f_j \ni \sigma \}$ the set of (indices of) flaws present in~$\sigma$.  (Here and elsewhere, we shall blur the distinction between flaws and their indices.)  

We consider algorithms which, in each flawed state~$\sigma\in\Omega^*$, choose a flaw~$f_i$ present in $\sigma$, i.e., $f_i \ni \sigma$, and attempt to leave (or ``fix'') $f_i$ by moving to a new state $\tau$ 
selected with probability $\rho_i(\sigma,\tau)$; we refer to such an attempt as {\it addressing}~$f_i$.  
We make minimal assumptions about how the algorithm choses which flaw to address at each step; e.g., it will be enough for the algorithm to choose the flaw with lowest index according to some fixed permutation.  
We say that a transition $\sigma \to \tau$, made to address flaw $f_i$, \emph{introduces} flaw $f_j $ if $\tau  \in f_j$, and either $\sigma \notin f_j $ or $j=i$.

For an arbitrary state $\tau \in \Omega$, flaw $f_i$, and set of flaws $S$, let 
\begin{equation}\label{eq:lalalakis}
\mathrm{In}_i^S(\tau)  :=  \{ \sigma \in f_i : \text{the set of flaws introduced by the transition $\sigma \to \tau$ includes $S$} \}  .
\end{equation}

For any fixed probability distribution $\mu>0$ on~$\Omega$ (either inherited from an application of the probabilistic method, as in the classical LLL, or introduced by the algorithm designer), we define the \emph{charge}{} of the pair $(i,S)$ with respect to $\mu$ to be
\begin{align}\label{charge_def}
\gamma_i^S := \max_{\tau \in \Omega} \left\{\frac{1}{\mu(\tau)}\sum_{ \sigma \in \mathrm{In}_i^S(\tau)} \mu(\sigma) \rho_i(\sigma,\tau )\right\} .
\end{align} 
That is, the charge $\gamma_i^S$ is an upper bound on the ratio between the ergodic flow into a state via transitions that introduce every flaw in $S$ (and perhaps more), and the probability of the state under~$\mu$. 

Our condition may now be stated informally as follows:
\begin{theorem}~\label{thm:softprelim}
%Let $M$ be any $|\Omega^*| \times | \Omega^*|$ invertible matrix. Let $\|\cdot\|$ be any operator norm. For $i \in [m]$ and $S \subseteq [m]$, let $\gamma_i^S := \| M A_i^S M^{-1} \|$. 
If there exist positive real numbers $\{\psi_i\}_{i=1}^m$ such that, for all $i \in [m] $,
\begin{align}\label{eq:softprelim}
   \frac{ 1 }{ \psi_i }   \sum_{  \substack{ S \subseteq [m]  } }  \gamma_i^S\prod_{ j \in S} \psi_j    \enspace  < 1  , 
\end{align}
then a local search algorithm as above reaches a flawless object quickly with high probability.
\end{theorem} 
\par\noindent
The phrase ``quickly with high probability" essentially means that the running time has expectation linear in the number of flaws and an exponential tail; we spell this out more formally in Section~\ref{sec:results}.

A key feature of Theorem~\ref{thm:softprelim} is the absence of a causality/dependency graph, present in all previous LLL conditions. This is because considering {\it point-to-set correlations}, i.e., how each flaw interacts with every other set of flaws, frees us from the traditional view of dependencies between individual events.
In our condition, every flaw may interact with every other flaw, as long as the interactions are sufficiently weak.  Notably, this is achieved without  any loss when specialized to the traditional setting
of a causality/dependency graph. To see this, note that if $S$ contains any flaw that is never introduced by addressing flaw~$f_i$, then $\gamma_i^S = 0$.  Thus, in the presence of a causality/dependency graph, the only terms contributing to the summation in~\eqref{eq:softprelim} are those that correspond to subsets of the graph neighborhood of flaw $f_i$, recovering the traditional setting.  

Besides relaxing the traditional notion of dependence, our condition is also quantitatively more powerful, even in the traditional setting. To get a feeling for this, we observe that the previously most powerful algorithmic LLL condition, due to Achlioptas, Iliopoulos and Kolmogorov~\cite{AIK}, can be derived from~\eqref{eq:softprelim} by replacing~$\gamma_i^S$ by $\gamma_i^{\emptyset}$ for every $S$ (and restricting~$S$ to subsets of the neighborhood of flaw~$f_i$, as discussed in the previous paragraph).  Since the charges $\gamma_i^S$ are decreasing in $S$, replacing $\gamma_i^{\emptyset}$ with $\gamma_i^S$ can lead to a dramatic improvement.  For example, if the flaws in $S$ are never introduced simultaneously when addressing flaw $f_i$, then $\gamma_i^S = 0$ and $S$ does not contribute to the sum in~\eqref{eq:softprelim}; in contrast, $S$ contributes $\gamma_i^{\emptyset}$, i.e., the maximum possible charge, to the corresponding sum in~\cite{AIK}. For a more detailed discussion of how our new condition subsumes existing versions of the LLL see Appendix~\ref{comparison}.

A natural question is whether Theorem~\ref{thm:softprelim} can be improved by replacing the word ``includes" with the word ``equals" in~\eqref{eq:lalalakis}, thus shrinking the sets $\mathrm{In}_i^S(\tau)$. The short answer is ``No," i.e., such a change invalidates the theorem. The reason for this is that in resampling algorithms we must allow for the possibility that flaws introduced when addressing~$f_i$ may later be fixed ``collaterally,'' i.e., as the result of addressing other flaws\, rather than by being specifically addressed by the algorithm. While it may seem that such collateral fixes cannot possibly be detrimental, they are problematic from an analysis perspective
as they can potentially increase the intensity of correlations between flaw $f_i$ and~$S$. Perhaps more convincingly, tracking collateral fixes and taking them into account also appears to be a bad idea in practice~\cite{walksat, walk_noise, balint2012choosing,probsat}: for example, local search satisfiability algorithms that select which variable to flip (among those in the targeted violated clause) based \emph{only} on which clauses will become violated, fare much better than algorithms that weigh this damage against the benefit of the collaterally fixed clauses.

Motivated by the above considerations, we introduce the notion of \emph{primary} flaws. These are flaws which, once present, can only be eradicated by being addressed by the algorithm, i.e., they cannot be fixed collaterally. Primary flaws allow us to change the definition of the sets $\mathrm{In}_i^S(\tau)$ in the desired direction. Specifically, say that a set of flaws $T$ \emph{covers} a set of flaws $S$ if:
\begin{enumerate}
\item
the set of primary flaws in $T$ \emph{equals} the set of primary flaws in $S$; and
\item
the set of non-primary flaws in $T$ \emph{includes} the set of non-primary flaws in $S$.
\end{enumerate}
In other words, we demand equality at least for the primary flaws.

\begin{theorem}\label{thm:soft}
Theorem~\ref{thm:softprelim} continues to hold if\/ $\mathrm{In}_i^S(\tau)$ is redefined by replacing
``includes" by ``covers"  in equation~\eqref{eq:lalalakis}.
\end{theorem}

The notion of primary flaws is one of our main conceptual contributions.  
Crucially for our applications, backtracking steps always introduce only primary flaws and thus, for such steps, we achieve an ideal level of control. 
The full version of our new algorithmic LLL condition, incorporating primary flaws, is spelled out formally in Theorem~\ref{soft} in Section~\ref{sec:results}.

\subsection{Technical overview: the Lov\'{a}sz Local Lemma as a spectral condition}\label{sec:spectral}

We conclude this introduction by sketching the techniques we use to prove our convergence criterion.  
As mentioned earlier, our main insight is to interpret LLL-inspired
convergence arguments for local search algorithms as methods for bounding the 
{\it spectral radius\/} of an associated matrix.
%We explain this viewpoint in the simple setting of analyzing
%the Moser-Tardos algorithm for $k$-SAT; later, in Section~\ref{soft_LLL}, we will employ it in
%a much more sophisticated form to derive our new convergence criterion.

As above, let $\Omega$ be a (large) finite set of states and let $\Omega^* \subseteq \Omega$ be the ``bad" part of~$\Omega$, comprising the flawed states.  Imagine a particle trying to escape~$\Omega^*$ by following a Markov chain on~$\Omega$ with transition matrix~$P$. Our task is to develop conditions under which the particle eventually escapes, thus establishing in particular that $\Omega^* \neq \Omega$.  Letting $\Amat$ be the $|\Omega^*| \times | \Omega^*|$ submatrix of $P$ that corresponds to transitions from $\Omega^*$ to $\Omega^*$, and $\Amatcomp$ the submatrix that corresponds
to  transitions from $\Omega^*$ to $\Omega \setminus \Omega^*$, we may, after a suitable
permutation of its rows and columns, write $P$ as:
\begin{align*}
P = \left[
\begin{array}{c|c}
\Amat & \Amatcomp \\
\hline
0 & I
\end{array}
\right]  .
\end{align*}
Here $I$ is the identity matrix, since we assume that the particle stops after reaching a flawless state. 

Let $\theta = [ \theta_1 \mid \theta_2 ]$ be the row vector corresponding to the probability distribution of the starting state, where $\theta_1$  and $\theta_2$ are the vectors that correspond to states in $\Omega^*$ and $\Omega \setminus \Omega^*$, respectively. Then, the probability that after $t$ steps the particle is still inside $\Omega^*$ is exactly $\| \theta_1 \Amat^t \|_1$. 
Therefore, for any initial distribution $\theta$, the particle escapes $\Omega^*$ if and only if the spectral radius, $\rho(\Amat)$, of $\Amat$ is strictly less than~1. Moreover, the rate of convergence is dictated by $1 - \rho(\Amat)$. Unfortunately, since $\Amat$ is huge and defined implicitly by an algorithm, the magnitude of its largest eigenvalue, $\rho(\Amat)$, is not readily available. 

To sidestep the inaccessibility of the spectral radius $\rho(\Amat)$, one can instead bound some {\it operator norm\/} $\| \cdot \|$  of~$\Amat$ and appeal to the
fact that $\rho(\Amat) \le \| \Amat \|$ for any such norm. 
Moreover, instead of bounding an operator norm of $A$ itself, one often first performs a ``change of basis" $\Amat' = M A M^{-1}$ for some invertible matrix~$M$ and then bounds $\|A'\|$, justified by the fact that $\rho(\Amat) = \rho(\Amat') \le \|A'\|$. The purpose of the change of basis  is to cast $\Amat$ ``in a good light" in the eyes of the chosen operator norm, in the hope of minimizing the cost of replacing the spectral norm with the operator norm.

As we explain in Appendix~\ref{app:spectral}, essentially all known analyses of LLL-inspired local search algorithms can be recast in the above framework of matrix norms. More significantly, this viewpoint allows the extension of such analyses to a wider class of algorithms. In Section~\ref{soft_LLL}, we will use this linear-algebraic machinery to prove our new convergence condition. The role of the measure~$\mu$ will be reflected in the change of basis~$M$, while the charges $\gamma^S_i$ will correspond to norms $\|M A^S_i M^{-1}\|_{1}$, where the $A^S_i$ are submatrices of the transition matrix~$A$.

\subsection{Follow-up work} 

Following the appearance of the conference version of this paper,
Davies et al.~\cite{davies2020algorithmic} used our new algorithmic LLL condition in their general framework for coloring locally sparse graphs.  Among other applications, they generalize our Theorem~\ref{improvement_mike} so that its conclusion holds under an analogous bound for the local $k$-cycle-density of the input graph, for any fixed $k \ge 3$. (Thus, when $k =3$ their result reduces to Theorem~\ref{improvement_mike}.) In a companion paper~\cite{davies2020graph},  the same authors also managed to improve the bound in our Theorem~\ref{sparse_graphs} by a factor of~$2$, but this result only implies an algorithm whose running time is exponential in the maximum degree of the graph.  An interesting feature of these papers is that they apply the LLL in sophisticated probability spaces that are inspired by so-called ``hardcore distributions" on independent sets.

A number of papers in the algorithmic LLL literature~\cite{Haeupler_jacm,NewBoundsHarris,EnuHarris,LLLWTL,kolmofocs} analyze properties of local search algorithms besides fast convergence (such as parallelization, bounds on the entropy of the output distribution, etc.).  It is known that such properties cannot be established for general local search algorithms~\cite{HV}, but it has been observed~\cite{kolmofocs,LLLWTL} that they do hold for a large class of so-called {\it commutative\/} algorithms. Very recently, and inspired by our work in this paper, Harris, Iliopoulos and Kolmogorov~\cite{matrix_comm} introduced a very natural and more general notion of commutativity (essentially matrix commutativity) that allows them to show a number of new refined properties of LLL-inspired local search algorithms with significantly simpler proofs.

Finally, Kolmogorov~\cite{kolmo_new} extended our new algorithmic LLL condition, showing that it can be combined with  well-known improvements of the standard Lov\'{a}sz Local Lemma for locally dense dependency graphs~\cite{bissacot,AIK}. However, we do not currently know a natural setting where this new condition applies, and Kolmogorov mentions that he views the significance of his contribution as mainly theoretical. (See also Appendix~\ref{comparison} for further discussion.)

\section{Statement of results}\label{sec:results}

\subsection{A new algorithmic LLL condition}\label{subsec:results:main}

Below we state our main result, which is a formal version of Theorem~\ref{thm:soft} discussed in Section~\ref{sec:informal}.

Recall that $\Omega$ is a finite set, $F = \{f_1, f_2, \ldots, f_m \}$ is a collection of subsets of $\Omega$ which  we refer to as {\it flaws}, and $\bigcup_{i \in [m]} f_i = \Omega^*$.  Our goal is to find a flawless object, i.e., an object in $\Omega \setminus \Omega^*$. For a state $\sigma$, we denote by $U(\sigma) = \{ j \in [m] : f_j \ni \sigma \}$ the set of (indices of) flaws present in~$\sigma$.  (Here and elsewhere, we shall blur the distinction between flaws and their indices.)  We consider algorithms that start in a state sampled from a probability distribution $\theta$ and, in each flawed state~$\sigma\in\Omega^*$, choose a flaw~$f_i\in U(\sigma)$, and attempt to leave (``fix'') $f_i$ by moving to a new state $\tau$  selected with probability $\rho_i(\sigma,\tau)$.  We refer to an attempt to fix a flaw, successful or not, as \emph{addressing} it. 
%We make minimal assumptions about how the algorithm choses which flaw to address at each step; e.g., it will be enough for the algorithm to choose the flaw with lowest index according to some fixed permutation.  
We say that a transition $\sigma \to \tau$, made to address flaw $f_i$, \emph{introduces} flaw $f_j \in U(\tau)  $ if $f_j  \notin U(\sigma)$ or    if $j=i$.  (Thus, a flaw (re)introduces itself when a transition fails to address it.)

%
%
%For a state $\sigma$, we denote by $U(\sigma) = \{ j \in [m] : f_j \ni \sigma \}$ the set of (indices of) flaws present in~$\sigma$.  (Here and elsewhere, we shall blur the distinction between flaws and their indices.) 
% We consider algorithms which, in each flawed state~$\sigma\in\Omega^*$, choose a flaw~$f_i$ in $U(\sigma)$ and attempt to leave (``fix'') $f_i$, by moving to a new state $\tau$ according to a probability distribution $\rho_i(\sigma,\cdot)$. 
%   We say that a transition $\sigma \to \tau$, made to address flaw $f_i$, \emph{introduces} flaw $f_j \in U(\tau)$ if $f_j \notin U(\sigma)$ or if $j=i$. (Thus, a flaw (re)introduces itself when a transition fails to remove it.)

Recall that $\theta$ denotes the probability distribution of the starting state. We denote by $\mathrm{Span}(\theta)$ the set of flaw indices that may be present in the initial state, i.e., $\mathrm{Span}(\theta) = \bigcup_{\sigma \in \Omega: \theta(\sigma) >0 } U(\sigma)$.

Let $\pi$ be an arbitrary permutation over  $[m]$. We say that an algorithm follows the \emph{$\pi$-strategy} if at each step the flaw it chooses to address is the one corresponding to the element of $U(\sigma)$ of lowest index according to $\pi$.

We now formalize the definitions of primary flaws and charges introduced informally in the introduction.

 \begin{definition}
A flaw $f_i$ is \emph{primary} if for every $\sigma \in f_i$ and every $j \neq i$, addressing $f_j$ at $\sigma$ always results in some $\tau \in f_i$, i.e., $f_i$ is never eradicated collaterally. For a given set $S \subseteq [m]$, we write $S^P$ and $S^N$ to denote the indices that correspond to \primary~and non-\primary~flaws in $S$, respectively.  
\end{definition}

%\marginpar{\tiny AS: Need to formally define ``covers" here!  Also, keep referring back to
%intro where appropriate}
\begin{definition}
We say that a set of flaws $T$ \emph{covers} a set of flaws $S$ if $T^P = S^P$ and $T^N \supseteq S^N$.
\end{definition}

\begin{definition}\label{sparser}
%For a state $\tau \in \Omega$, a flaw $f_i$, and a set of flaws $S$, let $\mathrm{In}_i^S(\tau)$ denote the set of states $\sigma $ such that there exists a transition $\sigma \to \tau$ in which the algorithm addresses flaw $f_i$ and the set of primary flaws introduced by the transition $\sigma \to \tau$ \emph{equals} $S^P$ and the set of non-primary flaws introduced by  $\sigma \to \tau$ contains $S^N$. 
%the set of flaws introduced contain $S$. 
%For every $i \in [m]$ and every set of flaw indices $S\subseteq [m]$,  let $A_i^S$ be the $|\Omega|\times|\Omega|$ matrix where $A_i^S[\sigma,\sigma'] = \rho_i(\sigma,\sigma')$ if the set of primary flaws introduced by the transition $\sigma \to \sigma'$ \emph{equals} $S^P$ and the set of non-primary flaws introduced by  $\sigma \to \sigma'$ contains $S^N$; otherwise $A_i^S[\sigma,\sigma'] = 0$.
For a state $\tau \in \Omega$, flaw $f_i$, and set of flaws $S$, let $$
%\begin{equation}\label{eq:lalalakis2}
\mathrm{In}_i^S(\tau)  =  \{ \sigma \in f_i : \text{the set of flaws introduced by the transition $\sigma \to \tau$ covers $S$} \}   . $$
%\end{equation}
Let $\mu>0$ be an arbitrary measure on $\Omega$. For every $i \in [m]$ and $S \subseteq [m]$, the \emph{charge} of $(i,S)$ with respect to $\mu$ is,
\begin{align}\label{formal_charge_definition} 
\gamma_i^S =  \max_{ \tau \in \Omega} 
	\left\{\frac{1} { \mu(\tau)  }
		\sum_{ \sigma \in  \mathrm{In}_i^S(\tau)}  \mu(\sigma) \rho_i(\sigma,\tau)   
	\right\}   . 
\end{align}
\end{definition}

%\begin{remark}
%In Theorem~\ref{thm:softprelim}, we used matrices where $A_i^S[\sigma,\sigma'] = \rho_i(\sigma,\sigma')$ if the set of flaws introduced by the transition $\sigma \to \sigma'$ contained $S= S^P \cup S^N$. The sparsification amounts to zeroing out all entries for which the set of primary flaws introduced is a strict superset of $S^P$. In particular, if $S^P = \emptyset$,  then all  entries corresponding to transitions that  introduce primary flaws are zeroed-out.
%\end{remark}

%For each  flaw index $i \in [m]$ and for each subset $S \subseteq [m] $, let $A_i^S$   be the $|\Omega| \times |\Omega| $ matrix  defined by $A_i^S[ \sigma,\sigma'] = \rho_i(\sigma,\sigma') $ if $\sigma \in f_i$, $S^N \subseteq U^N(\sigma') \setminus \left(  U^N(\sigma)  \setminus \{ i \}  \right)$ and $S^P  = U^P(\sigma') \setminus \left(  U^P(\sigma)  \setminus \{ i \}  \right)$ , and $A_i^S[\sigma,\sigma'] = 0$ otherwise. 

%For a state $\sigma$, let $e_{\sigma }$ denote  the indicator vector of $\sigma$, i.e., $e_{\sigma}[\sigma] = 1$ and  $e_{\sigma}[\tau] = 0 $ for all $\tau \in \Omega \setminus \{ \sigma \}$. 

We now state the formal version of our main result, Theorem~\ref{thm:softprelim} of the introduction.

\begin{theorem}[Main Result]~\label{soft}
If there exist positive real numbers $\{\psi_i\}_{i \in [m] }$ such that, for every $i \in [m] $,
\begin{align}
\zeta_i :=   \frac{ 1 }{ \psi_i }  \sum_{  \substack{ S \subseteq [m]  } }  \gamma_i^S\prod_{ j \in S} \psi_j     \enspace  < 1  , \label{eq:main_soft_condition}
\end{align}
then, for every permutation $\pi$ over $[m]$, the probability that an algorithm following the $\pi$-strategy fails to reach a flawless state within $(T_0 + s)/\delta$ steps is $2^{-s}$, where $\delta = 1 - \max_{i \in [m]} \zeta_i $, and
\begin{align*}
 T_0 & = 
%\log_2 \left(  \max_{\sigma \in \Omega} \frac{  \theta(\sigma) }{ \mu(\sigma) } \right)+   \log_2 \biggl( \sum_{S \subseteq  \mathrm{Span}(\theta)} \prod_{j \in S} \psi_j  \biggr) + \log_2  \biggl(  \max_{S \subseteq [m] } \frac{1}{ \prod_{j \in S} \psi_j} \biggr) \\
\log_2  \mu_{\min}^{-1}+ m \log_2\left(\frac{1+\psi_{\max}}{\psi_{\min}}\right) ,
\end{align*} 
with $\mu_{\min} = \min_{\sigma \in \Omega} \mu(\sigma) $, $\psi_{\max}  = \max_{i \in [m] } \psi_i $ and $\psi_{\min} = \min_{i \in [m]}  \psi_i $.  
\end{theorem}

\begin{remark} 
In typical applications, $\mu$ and $\{\psi_i\}_{i \in [m] }$ are such that $T_0 = O(\log |\Omega| + m)$ and the sum in~\eqref{eq:main_soft_condition} is easily computable, as $\gamma_i^S = 0$ for the vast majority of subsets $S$. 

\end{remark}

%
%\begin{remark} 
%The requirement $\sum_{\sigma \in \Omega} \| M e_{\sigma} \|  = 1$ is not really necessary. We
%impose it because in applications $M$ is typically diagonal with positive entries, in which case the normalization $\sum_{\sigma \in \Omega} \| M e_{\sigma}  \| = 1$ simplifies the expressions for the running time.
%\end{remark}

\begin{remark}\label{priority_remark}
For any fixed  permutation $\pi$, the charges $\gamma_i^S$ can be reduced by removing from $\mathrm{In}_i^S(\tau)$ every state for which $i$ is not the lowest indexed element of $U(\sigma)$ according to $\pi$.
\end{remark}

\begin{remark}
Theorem~\ref{soft} also holds for algorithms using flaw choice strategies other than $\pi$-strategies. We discuss some such strategies in Section~\ref{remarks_proofs_strategies}. However, there is good reason to expect that it does \emph{not} hold for \emph{arbitrary} flaw choice strategies (see~\cite{kolmofocs}).
\end{remark}

Finally, we state a refinement of our running time bound that will be important in order to get the best convergence guarantees in the applications of pure backtracking algorithms in Section~\ref{SoftcoreLLL}.
\begin{remark}\label{cor:faster_back}
The upper bound on $T_0$ in Theorem~\ref{soft} can be replaced
%\marginpar{\tiny AS: Make sure pf of Main Theorem clearly shows proof of both simplified
%bound above and refined bound of Rem 2.4}
%\marginpar{\tiny AS: Merge Rem 2.4 \& Cor 2.5 into a single Remark (connected by ``moreover").
%Remove last sentence from Cor 2.4(?).
%Say outside Remark that this is important in order  to get the best running time in the applications of pure backtracking algs in Section 5}
by the more refined bound:
\begin{align*}
 T_0 & = 
\log_2 \left(  \max_{\sigma \in \Omega} \frac{  \theta(\sigma) }{ \mu(\sigma) } \right)+   \log_2 \biggl( \sum_{S \subseteq  \mathrm{Span}(\theta)} \prod_{j \in S} \psi_j  \biggr) + \log_2  \biggl(  \max_{S \subseteq [m] } \frac{1}{ \prod_{j \in S} \psi_j} \biggr)   .
\end{align*} 
Moreover, if (as in pure backtracking algorithms) every flaw is primary, and (as is typical in pure backtracking algorithms) every flaw is present in the initial state, and if $\psi_i\in (0,1]$ for all~$i$,
then $T_0 = \log_2  \mu_{\min}^{-1}$.
%\begin{align*}
%T_0 = \log_2  \mu_{\min}^{-1}  \enspace. 
%\end{align*}
\end{remark}

%As we will see in Section~\ref{SoftcoreLLL}, in  pure backtracking algorithms every flaw is primary  leading to the following improvement in the running time bound of the algorithm, i.e., the value of $T_0$ in Theorem~\ref{soft}.  
%
%
%
%\blue{
%\begin{corollary}
%Let $\mathcal{I}(\theta)$ be the set comprising the sets of flaw-indices that may be present in a state selected according to $\theta$. If every flaw is primary, then the sum over $S \subseteq  \mathrm{Span}(\theta)$ in the definition of $T_0$ can be restricted to $S \in \mathcal{I}(\theta)$. 
%\end{corollary}
%}

\subsection{Application to graph coloring}
\label{subsec:results:coloring}

In graph coloring one is given a graph $G=(V,E)$ and the goal is to find a mapping of $V$ to a set of $q$ colors so that no edge in $E$ is monochromatic. The \emph{chromatic number}, $\chi(G)$, of $G$ is the smallest integer $q$ for which this is possible. 
%\marginpar{\tiny AS: State clearly somewhere before Thm that we've reparameterized $T$ to $f$}
Given a set $\mathcal{L}_v$ of colors for each vertex $v$ (called a {\it list\/}), a list-coloring maps each $v \in V$ to a color in $\mathcal{L}_v$ so that no edge in $E$ is monochromatic. A graph is $q$-list-colorable if it has a list-coloring no matter how one assigns a list of $q$ colors to each vertex. The \emph{list chromatic number}, $\chi_{\ell}(G)$, is the smallest  $q$ for which $G$ is $q$-list-colorable. Clearly $\chi_\ell(G)\ge\chi(G)$. A celebrated result of Johansson~\cite{JO} established that there exists a large constant $C>0$ such that every {\it triangle-free\/} graph with maximum degree $\Delta \ge \Delta_0$ can be list-colored using $C \Delta/\ln \Delta$ colors. Very recently, using the entropy compression method, Molloy~\cite{molloy2017list} improved Johansson's result, replacing $C$ with $(1+\epsilon)$ for any $\epsilon > 0$ and all $\Delta \ge \Delta_{\epsilon}$. (Soon thereafter, Bernshteyn~\cite{MolloyLLL} established the same bound for the list chromatic number, non-constructively, via the Lopsided LLL, and Iliopoulos~\cite{LLLWTL} showed that the algorithm of Molloy can be analyzed using the algorithmic LLL condition of~\cite{AIK}, avoiding the need for a problem-specific entropy compression argument.)

Our first result related to graph coloring is a generalization of Molloy's result,  bounding the list-chromatic number as a function of the number of triangles in each neighborhood. 
Specifically, in Section~\ref{proof_improvement_mike} we establish the following theorem which is a key ingredient in the proof of Theorem~\ref{color_inform}.   Note that, in order to comply with the standard notation used in  results in the area, we express the bound on the number of triangles as $\Delta^2/f$; the triangle-free case then corresponds to $f=\Delta^2 +1$.  We stress that Molloy's proof breaks in the presence of even a single triangle per vertex.
%(the triangle-free case corresponding to $f=\Delta^2 +1$. \blue{We reparametrized $T$ to $f$ to comply with the standard notation used in  results in the area.}).
\begin{theorem}\label{improvement_mike}
Let $G$ be any graph with maximum degree $\Delta$ in which the neighbors of every vertex span at most $\Delta^2/f$ edges. For all $\epsilon > 0$, there exists $\Delta_{\epsilon}$ such that if $\Delta \ge \Delta_{\epsilon}$ and $  f \in [ \Delta^{\frac{2+2\epsilon}{1+2\epsilon} } (\ln \Delta)^2  , \Delta^2+1]$, then
\[
\chi_{\ell}(G) \le  (1 + \epsilon )  \Delta / \ln  \sqrt{f}  . 
\]
Furthermore, if $G$ is a graph on $n$ vertices then, for every $c> 0$,  there exists an algorithm that constructs such a coloring in polynomial time with probability at least $1 - \frac{1}{n^c} $.
\end{theorem}

Theorem~\ref{improvement_mike} is interesting for several reasons. First, random graphs suggest that it is sharp, i.e., that no efficient algorithm can color graphs satisfying the conditions of the theorem with $(1-\epsilon) \Delta / \ln  \sqrt{ f}$ colors.  More precisely, Proposition~\ref{Random_Graphs} below, proved in Appendix~\ref{random_graphs_proof}, implies that any such algorithm would entail coloring random graphs using fewer than twice as many colors as their chromatic number. This would be a major (and unexpected) breakthrough in random graph theory, where beating this factor of two has been an elusive goal for over 30 years. Besides the lack of progress, further evidence for the optimality of this factor of two is that it corresponds precisely to a phase transition in the geometry of the set of colorings~\cite{mitsaras_barriers}, known as the \emph{shattering threshold}. Second, Theorem~\ref{improvement_mike} establishes the existence of an algorithm that is robust enough to apply to worst-case graphs, while at the same time matching the performance of the best known (and highly tuned) algorithms for random graphs:

\begin{proposition}\label{Random_Graphs}
For every  $\epsilon > 0$ and $d \in  (d_{\epsilon}  \ln n ,  (n \ln n)^{ \frac{1}{3} }  )$, there exist $\Delta=\Delta(d, \epsilon)$ and $f = f(d, \epsilon)$ such that with probability tending to $1$ as $n \to \infty$, a random graph $G = G(n, d/n)$ satisfies the conditions of Theorem~\ref{improvement_mike} and $\chi(G) \ge (\frac{1}{2} - \epsilon) \Delta / \ln  \sqrt{ f}$.
\end{proposition}

Third, armed with Theorem~\ref{improvement_mike}, we are able to prove
the following result concerning the chromatic number of {\it arbitrary\/} graphs, as a function of the maximum degree and the maximum number of triangles in any neighborhood:
% by closely following the analysis of~\cite{alon1999coloring}.
\begin{theorem}\label{sparse_graphs}
Let $G$ be a graph with maximum degree $\Delta$ in which the neighbors of every vertex span at most $\Delta^2/f$ edges. For all $\epsilon > 0$, there exist $\Delta_{\epsilon}, f_{\epsilon}$ such that if $\Delta \ge \Delta_{\epsilon}$ and $f \in [f_{\epsilon}, \Delta^2+1]$, then 
\begin{equation}\label{eq:looks_good}
\chi(G) \le  (2+\epsilon) \Delta / \ln  \sqrt{ f}  .
\end{equation}
Furthermore, if $G$ is a graph on $n$ vertices then, for every $c> 0$,  there exists an algorithm that constructs such a coloring in polynomial time with probability at least $1 - \frac{1}{n^c} $.
\end{theorem}
Theorem~\ref{sparse_graphs} improves a classical result of Alon, Krivelevich and Sudakov~\cite{alon1999coloring} which established~\eqref{eq:looks_good} with an unspecified (large) constant in place of $2+\epsilon$.  
The main idea in their analysis is to break down the input graph into triangle-free subgraphs, and color each one of them separately using distinct sets of colors by applying the result of Johansson~\cite{JO}. Note that even if one  used Molloy's recent result~\cite{molloy2017list} in place of Johansson's in this scheme, the corresponding constant would still be in the thousands. Instead, we break down the graph into subgraphs with {\it few\/} triangles per neighborhood, and use Theorem~\ref{improvement_mike} to color the pieces. The proof of Theorem~\ref{sparse_graphs} can be found in Appendix~\ref{determ_graphs_proof}. 
%\red{Third, the value for $f$ for which Theorem~\ref{improvement_mike} applies essentially corresponds to threshold on the local triangle density up to which the sparsity of depth-one neighborhoods characterizes the chromatic number of random graphs.   }
%%Independently of whether Theorem~\ref{sparse_graphs} is tight or not, we find it remarkable that using only the assumption of sparse depth-one neighborhoods suffices to bring the chromatic number of arbitrary graphs within a factor of 2 of the chromatic number of random graphs and within parity in terms of efficient computation. 

As final remark, we note that Vu~\cite{vu2002general} proved the analogue of the main result of~\cite{alon1999coloring} (again with a large constant) for the list chromatic number. While we don't currently see how to sharpen Vu's result to an analogue of Theorem~\ref{sparse_graphs} for the list chromatic number using our techniques, we note that our Theorem~\ref{improvement_mike} improves over~\cite{vu2002general} for all $ f \ge   \Delta^{\frac{2+2\epsilon}{1+2\epsilon} } (\ln \Delta)^2$.

\section{Proof  of main theorem}\label{soft_LLL}

In Sections~\ref{witness_forests} and \ref{bounding_the_sum} we present the proof of our main result, Theorem~\ref{soft}. In Section~\ref{remarks_proofs_strategies} we show how to extend the theorem to allow flaw choice strategies other than following a fixed permutation over flaws.

Throughout this section we use standard facts about operator norms, summarized briefly in Appendix~\ref{background}.

\subsection{Charges as norms of transition matrices} 

We will first show how charges can be seen as the norms of certain transition matrices. For more concrete examples of this connection, see Appendix~\ref{app:spectral}.

Recall that for any $S \subseteq [m]$, we denote by $S^P$ and $S^N$ the subsets of $S$ that correspond to \primary\ and non-\primary~flaws, respectively. 
\begin{definition}
For every $i \in [m]$ and every set of flaw indices $S\subseteq [m]$,  let $A_i^S$ be the $|\Omega|\times|\Omega|$ matrix where $A_i^S[\sigma,\tau] = \rho_i(\sigma,\tau)$ if the set of flaws introduced by $\sigma \to \tau$ covers $S$, i.e., the set of primary flaws introduced by the transition $\sigma \to \tau$ \emph{equals} $S^P$ and the set of non-primary flaws introduced by  $\sigma \to \tau$ contains $S^N$; otherwise $A_i^S[\sigma,\tau] = 0$.
\end{definition}

Let $\|  \cdot \|_1$ denote the  matrix norm induced by the $L^1$-vector-norm, and recall that it is equal to the max column sum. Let also $M = \mathrm{diag}(\mu(\sigma))$ denote the $|\Omega| \times | \Omega |$ diagonal matrix whose entries correspond to the probability measure~$\mu$.  
Our key observation is that the charges~$\gamma_i^S$ introduced in~\eqref{formal_charge_definition}
can be expressed as
%Recall Definition~\ref{sparser} and~\eqref{formal_charge_definition}. Our key observation is that
\begin{align}\label{charg_norms}
\gamma_i^S = \|  M A_i^S M^{-1} \|_1 .
\end{align}
The reader is encouraged to verify this equivalence, which is an immediate consequence of the
definitions.
\begin{remark}\label{remark:norms}
Although we are specializing here to the $\| \cdot \|_1$ norm and matrix $M = \mathrm{diag}(\mu(\sigma))$, Theorem~\ref{soft} holds for any choice of matrix  norm and  invertible matrix $M$. It is an interesting research direction whether using other norms can be useful in applications. 
\end{remark}
 
\subsection{Tracking the set of current flaws} \label{witness_forests}

We say that a trajectory $\Sigma = ( \sigma_1, \sigma_2, \ldots, \sigma_{t+1} )$  followed by the algorithm  is a  \emph{bad $t$-trajectory} if  every state $\sigma_i$, $i \in [t+1]$, is flawed. Thus, our goal is to bound the probability that the algorithm follows a bad $t$-trajectory.

Given a bad trajectory, we will track the flaws introduced into the state in each step, where a flaw is  said to ``introduce itself" whenever addressing it fails to remove it. Of the flaws introduced at each step, we disregard those that later  get eradicated collaterally, i.e., by an action addressing some other flaw. The rest form the ``witness sequence" of the trajectory, i.e., a sequence of sets of flaws. 

Fix any permutation $\pi$ on $[m]$. For any $S\subseteq [m]$, let $\pi(S) = \min_{j \in S}\pi(j)$, i.e., the lowest index in $S$ according to $\pi$. Recalling that $U(\sigma)$ is the set of indices of flaws present in $\sigma$, in the following we assume that the index of the flaw addressed in state $\sigma$ is $\pi(U(\sigma))$, which we sometimes abbreviate as $\pi(\sigma)$. Also, to lighten notation, we will denote $A \setminus \{\pi(B)\}$ by $A - \pi(B)$.

\begin{definition}\label{BC}
Let $\Sigma = (
\sigma_1, \sigma_2, \ldots, \sigma_{t+1})
$ be any bad $t$-trajectory.  
%Let $B_i$ comprise the indices of the flaws ``introduced" by the $i$-th step, where we say that a flaw $f_i$  ``introduces itself" if it remains present in $\sigma_{i+1}$ after an action from $A(\red{[i]}, \sigma_{i})$ is taken. Formally,
%
Let  $B_0 = U(\sigma_1)$. For $1 \le i \le t$, let
%\[
%B_i = U(\sigma_{i+1}) \setminus ( U (\sigma_i) \setminus  
%I_
%{\pi}(U(
%\sigma_i)
%)) \enspace ,
%\]
\[
B_i = U(\sigma_{i+1}) \setminus [ U (\sigma_i) - \pi(\sigma_i)]  ,
\]
i.e., $B_i$ comprises the indices of the flaws introduced in the $i$-th step. For $0 \le i \le t$, let
%\[
%C_i = 	\{k \in B_i \mid  \exists j \in [i+1,t] :   k \notin U(\sigma_{j+1})  \wedge  \forall \ell \in [i+1,j]:   k \ne \pi(U(\sigma_{\ell})) \}
%\enspace ,
%\]
\[
C_i = 	\{k \in B_i \mid  \exists j \in [i+1,t] :   k \notin U(\sigma_{j+1})  \wedge  \forall \ell \in [i+1,j]:   k \ne \pi(\sigma_{\ell}) \}
 ,
\]
i.e., $C_i$ comprises the indices of the flaws introduced in the $i$-th step that get eradicated collaterally. The \emph{witness sequence} of~$\Sigma$ is the sequence of sets 
\[
w(\Sigma) =  %(B_i \setminus C_i )_{i=0}^t 
(B_0 \setminus C_0, B_1 \setminus C_1, \ldots, B_t \setminus C_t) . 
\]
\end{definition}

A crucial feature of witness sequences is that they allow us to recover the sequence of flaws addressed.

\begin{definition}\label{def:reconstruct}
Given an \emph{arbitrary} sequence $S_0, \ldots, S_t$ of subsets of~$[m]$,
%$(S_i)_{i=0}^t \subseteq [m]$, 
let $S^*_1 = S_0$, while for $1 \le i \le t$, let   
%\[
%S^*_{i+1} =
%\begin{cases}
%\left( S^*_i \setminus \pi(S^*_i)  \right) \cup S_i & \text{if $S^*_i \neq \emptyset$} \enspace , \\
%\emptyset	& \text{otherwise}  \enspace .
%\end{cases}
%\]
\[
S^*_{i+1} =
\begin{cases}
\left[ S^*_i - \pi(S^*_i)  \right] \cup S_i & \text{if $S^*_i \neq \emptyset$}  ; \\
\emptyset	& \text{otherwise}  .
\end{cases}
\]
If $S^*_i \neq \emptyset$ for all $1 \le i \le t$, then we say that $(S_i)_{i=0}^t$ is \emph{plausible} and write $\pi(S^*_i ) = (i)$. %\red{When the sequence $s$ is understood from context we will write $[i]$ instead of $s[i]$.}
\end{definition}

\begin{lemma}\label{lem:reconstruct}
If $\Sigma  = (\sigma_1, \sigma_2, \ldots, \sigma_{t+1})$ is any bad $t$-trajectory, then $w(\Sigma) = (S_0, \ldots, S_t)$ is plausible, $\pi(\sigma_i) = \pi(S^*_i ) = (i)$ for all $1 \le i \le t$, and for every flaw index $z \in [m]$, the number of times $z$ occurs in the multiset $\bigcup_{i=0}^t S_i$ minus the number of times it occurs in the multiset $\bigcup_{i=1}^t (i)$ equals $\bm{1}_{z \in S^*_{t+1}}$.
\end{lemma}
\begin{proof}
Recall that $S_i = B_i \setminus C_i$. For $1 \le i \le t+1$, let $L_i$ comprise the elements of $U(\sigma_i)$ eradicated collaterally during the $i$-th step, and let $H_i$ comprise the elements  of $U(\sigma_i)$ eradicated collaterally during any step $j \ge i$. Observe that $H_{i+1} = (H_i \setminus L_i) \cup C_i$. We will prove, by induction, that for all $1 \le i \le t+1$,
\begin{eqnarray}
S_i^* & \subseteq & U(\sigma_i) ;\label{eq:nomadness} \\
U(\sigma_i) \setminus S_i^* & = & H_i
%& \subseteq & \bigcup_{j=0}^{i-1} C_j 
 . \label{eq:onlychaff}
\end{eqnarray}
Note that if~\eqref{eq:nomadness} and \eqref{eq:onlychaff} hold for a given $i$, then $\pi(\sigma_i) = \pi(S_i^*)$, since $\pi(\sigma_i) \not\in H_i$  by the definition of $H_i$, and $\pi(A) = \pi(A \setminus B)$ whenever $\pi(A) \not\in B$. Moreover,  $S_i^* \neq \emptyset$, because otherwise $U(\sigma_i) = H_i$, an impossibility. To complete the proof it suffices to note that for any $z \in [m]$, the difference in question equals $\bm{1}_{z \in U(\sigma_{t+1})}$ and that $U(\sigma_{t+1}) = S_{t+1}^*$ since, by definition, $H_{t+1} = \emptyset$. The inductive proof is as follows. 

For $i=1$, \eqref{eq:nomadness} and \eqref{eq:onlychaff}  hold since $S^*_1 = B_0 \setminus C_0$, while $U(\sigma_1) = B_0$. If \eqref{eq:nomadness} and \eqref{eq:onlychaff}  hold for some $i \geq 1$, then $S^*_{i+1} = \left[ S^*_i - \pi(\sigma_i)  \right] \cup S_i$ while, by definition, $U(\sigma_{i+1}) = [ \left( U(\sigma_i) - \pi(\sigma_i) \right)\setminus L_i] \cup B_i$. Thus, the fact that $S_i^* \subseteq U(\sigma_i)$ trivially implies $S_{i+1}^* \subseteq U(\sigma_{i+1})$, while
\[
U(\sigma_{i+1}) \setminus S_{i+1}^* = ( ( U(\sigma_{i}) \setminus S_{i}^* ) \setminus L_i) \cup (B_i \setminus S_i) = (H_i \setminus L_i) \cup C_i = H_{i+1}  .
\]
\end{proof}

The first step in our proof of Theorem~\ref{soft} is to give an upper bound on the probability that a given witness sequence occurs in terms of the charges $\gamma_i^S$.  In particular, and in order to justify Remark~\ref{remark:norms}, we will use an arbitrary norm $\| \cdot \|$ and invertible matrix $M$. 

Recall that $\| \cdot \|_*$ denotes the \emph{dual} of norm $\| \cdot \|$ and let $\theta^{\top} \in [0,1]^{ | \Omega | } $ denote the row vector corresponding to the probability distribution of the initial state $\sigma_1$. Moreover, for a state $\sigma $, let $e_{\sigma} $ denote the indicator vector of $\sigma$, i.e., $e_{\sigma}[ \sigma ] = 1$ and $e_{\sigma}[ \tau] = 0$ for all $\tau \in \Omega \setminus \{ \sigma \}$. 

\begin{lemma}\label{issuing_charges_lemma}
Fix any integer $t \ge 0$ and let $\Sigma$ be the random variable $(\sigma_1, \ldots, \sigma_{t+1})$.
Fix any arbitrary invertible matrix $M$ and operator norm $\| \cdot\|$, and let $\lambda_i^S = \| M A_i^S M^{-1} \|$. For any plausible sequence $\phi = (S_0, \ldots, S_t)$,
\begin{align}\label{weighted_sum_backtracking}
\Pr[w(\Sigma) = \phi]
\le
\| \theta^{\top} M^{-1} \|_*   \left(   \sum_{ \tau \in \Omega} \| M e_{\tau}   \| \right) \prod_{i = 1}^{t} \lambda_{(i)}^{S_i}   .
\end{align}
\end{lemma}
\begin{proof}
% 
%By Definition~\ref{BC} and Lemma~\ref{lem:reconstruct}, a necessary condition for $w(\Sigma) = \phi$ to occur is that $(i) \in U(\sigma_i)$ and $S_i \subseteq B_i$, for every $1 \le i \le t$.
%
%
%\begin{align}\label{a_crude_bound}
%\Pr[ \text{$\phi$ occurs} \wedge \sigma_{t+1} = \tau ]  \le  \Pr\left[ \bigwedge_{i=1 }^t   \left( [i]  \in U(\sigma_i)  \right)  \bigwedge_{i=1}^{t} \left(  S_{i} \subseteq  B_i \right)  \bigwedge \sigma_{t+1} = \tau \right]  \enspace.
%\end{align} 
%\begin{align}\label{a_crude_bound}
%%\Pr[ \text{$\phi$ occurs} \wedge \sigma_{t+1} = \tau ]  \le  \Pr\left[ 
%\bigwedge_{i=1 }^t   
%\left( 
%\phi(i)  \in U(\sigma_i)  
%\right)  
%\bigwedge_{i=1}^{t} 
%\left(  
%S_{i} \subseteq  B_i(\Sigma) 
%\right)  
%%\bigwedge \sigma_{t+1} = \tau \right]  
%\enspace.
%\end{align} 
%\begin{align}%\label{a_crude_bound}
%%\Pr[ \text{$\phi$ occurs} \wedge \sigma_{t+1} = \tau ]  \le  \Pr\left[ 
%\bigwedge_{i=1 }^t   
%\left(
%[i]  \in U(\sigma_i)
%\wedge
%S_{i} \subseteq  B_i(\Sigma) 
%\right)
%%\bigwedge \sigma_{t+1} = \tau \right]  
%\enspace.
%\end{align} 
%
%
%
%
Recall that for any $S \subseteq [m]$, we denote by $S^P$ and $S^N$ the subsets of $S$ that correspond to \primary\ and non-\primary~flaws, respectively. By Definition~\ref{BC} and Lemma~\ref{lem:reconstruct}, a necessary condition for $w(\Sigma) = \phi$ to occur is that $(i) \in U(\sigma_i)$ and $S_i \subseteq B_i$, for every $1 \le i \le t$.  Moreover, since \primary~flaws are never eradicated collaterally, i.e., $C_i^P = \emptyset$ always, it must also be that $S_i^P = B_i^P$ for $1 \le i \le t$. 
Fix any state $\tau \in \Omega$. The  probability that 
$(1) \in U(\sigma_1)   %\right
\wedge 
S_{1}^P =%U(\tau) \setminus \left(U(\sigma) \setminus \{ i \} \right) 
B_1^P(\Sigma)
\wedge
S_{1}^N \subseteq 
%U(\tau) \setminus \left(U(\sigma) \setminus \{ i \} \right) 
B_1^N(\Sigma) \wedge %\left( 
\sigma_2  = \tau $ 
equals the $\tau$-column (coordinate) of the row-vector $\theta^{\top} A_{(1)}^{S_{1}}$. More generally, we see that for any $t\ge 1$,
\begin{align}
\Pr\left[ \bigwedge_{i=1 }^t   \left( (i)  \in U(\sigma_i)  \right)  \bigwedge_{i=1}^{t} \left( S_{i}^P = B_i^P \right)   \bigwedge_{i=1}^{t} \left(  S_i^N \subseteq B_i^N \right)      \bigwedge \sigma_{t+1} = \tau \right]  = \theta^{\top} \prod_{i=1}^t A^{S_{i}}_{(i)} e_{\tau} .
\end{align} 
Consider now any  vector norm $\|  \cdot \|$ and the corresponding operator norm. By~\eqref{basic_inequality},
\begin{align}\label{eq:meat}
 \theta^{\top} \prod_{i=1}^t A_{(i)}^{S_{i}} e_{\tau} = 
\theta^{\top} M^{-1} \left(\prod_{i=1}^t M A^{S_{i}}_{(i)}M^{-1}\right) M e_{\tau}
& \le   
\left|\left|
\theta^{\top} M^{-1} \left(\prod_{i=1}^t M A^{S_{i}}_{(i) }M^{-1}\right) \right|\right|_{*}  \| M e_{\tau}  \|  .
\end{align}
Summing~\eqref{eq:meat} over all $\tau \in \Omega$  we conclude that
\begin{align}\label{vectors_vs_matrices}
\Pr[ w(\Sigma)=\phi] =
\sum_{\tau \in \Omega} \Pr[ w(\Sigma)=\phi \wedge \sigma_{t+1} = \tau ]  
\le
 \left|\left|    \theta^{\top} M^{-1} \prod_{i=1}^t   M A^{S_{i}}_{(i) }M^{-1} \right|\right|_{*}   \sum_{\tau \in \Omega }  \|M  e_{ \tau } \|.
\end{align}
Applying~\eqref{submult_1} and then~\eqref{submult_2} to~\eqref{vectors_vs_matrices}, and recalling the definition of $\lambda^{S_{i}}_{(i)}$, we conclude that
\[
\Pr[ w(\Sigma)=\phi] \le   \| \theta^{\top} M^{-1}       \|_* \left( \sum_{\tau \in \Omega }  \|M  e_{ \tau } \| \right)    \prod_{i=1}^t  \| M A^{S_{i}}_{(i) } M^{-1}   \|  
= \| \theta^{ \top} M^{-1} \|_*    \left( \sum_{\tau \in \Omega }  \|M  e_{ \tau } \| \right)   \prod_{i=1}^t \lambda^{S_{i}}_{(i) }  	 ,
\]
as claimed.
\end{proof}

Now define $$
  \mathcal{F}_t = \{w(\Sigma) : \text{$\Sigma$ is a bad $t$-trajectory of the algorithm}\}.  $$ 
Since $\mathcal{F}_t$ contains only plausible sequences, an immediate  corollary of Lemma~\ref{issuing_charges_lemma} is a bound on the probability that the algorithm fails in $t$ steps.
\begin{corollary}\label{corollary_basic}
The probability  that the algorithm fails to reach a flawless state within $t$ steps is at most
\begin{align}\label{final_step}
\left( \max_{ \sigma \in \Omega } \frac{ \theta(\sigma) }{ \mu(\sigma)  } \right)   \cdot \sum_{ \phi \in \mathcal{F}_t} \prod_{i=1}^t \gamma^{S_i}_{(i)} .
\end{align}
\end{corollary}
\begin{proof}
We apply Lemma~\ref{issuing_charges_lemma} with $M = \mathrm{diag}(\mu(\sigma) )$ and the $\| \cdot \|_1$-norm.  Since the dual norm of $\|  \cdot \|_1$ is $\|  \cdot \|_{\infty}$, we have $\|  \theta^{\top} M^{-1} \|_{\infty} = \max_{ \sigma \in \Omega } \frac{ \theta(\sigma) }{ \mu(\sigma)  }$. Combining this with the fact that $\sum_{ \tau \in \Omega } \| M e_{\tau} \|_1 = 1$ concludes the proof.
\end{proof}

Thus, to complete the proof of Theorem~\ref{soft} we are left with the task of bounding the sum in~\eqref{final_step}.

\subsection{Bounding the sum}\label{bounding_the_sum}

%Specifically, we introduce a branching process that produces every labeled forest in $\mathcal{F}_t$ with positive probability and we use it to bound the sum $\sum_{\phi \in \mathcal{F}_t} \prod_{i = 1}^{t} \gamma_{S_{i}}^{[i]} $. 
Given $\psi_1, \ldots, \psi_m>0$ and $S \subseteq [m] $, let  $\pot(S) = \prod_{j \in S}   \psi_{j}$, with $\pot( \emptyset) = 1$. 
%\end{equation}
For each $i \in [m]$, let
\[
\zeta_i = \frac{1}{\psi_i} \sum_{S \subseteq [m]} \gamma^S_i \pot(S)  . 
\]
Finally, for each $i \in [m]$ consider the distribution on $2^{[m]}$ that assigns to each $S \subseteq [m]$ the probability
\[
p(i,S) := \frac{\gamma^S_i \pot(S)}{\sum_{S \subseteq [m]} \gamma^S_i \pot(S)} =
\frac{\gamma^S_i \pot(S)}{\zeta_i \psi_i} .
\]
For any $S_0 \subseteq [m]$, let $\mathcal{F}_t(S_0)$ comprise the witness sequences in $\mathcal{F}_t$ whose first set is $S_0$.  Consider the probability distribution on sequences of subsets of $[m]$ generated as follows: $R_1 = S_0$; for $i \ge 1$, if $R_i \neq \emptyset$, then $R_{i+1} = (R_{i} - \pi(R_i)) \cup S_i$, where $\Pr[S_i = S] = p(\pi(R_i), S)$ for any $S \subseteq [m]$. Under this distribution, by Lemma~\ref{lem:reconstruct}, each $\phi =(S_0, \ldots, S_t) \in  \mathcal{F}_t(S_0)$ receives probability $p_{\phi} = \prod_{i=1}^t p((i),S_i)$, while $\sum_{\phi \in \mathcal{F}_t(S_0)} p_{\phi} \le 1$. At the same time, by the last claim in Lemma~\ref{lem:reconstruct},
%for every flaw index $z \in [m]$ and $j \ge 1$, the number of times $z$ occurs in the multiset $\bigcup_{i=0}^j S_i$ minus the number of times $z$ occurs in the multiset $\bigcup_{i=1}^j [i]$ belongs in $\{0,1\}$ and, in fact, equals the indicator of whether $z \in S^*_{j+1}$. Therefore,
\begin{equation}\label{eq:kolpa}
p_{\phi} = \prod_{i=1}^t p((i),S_i) = 
\left( \prod_{i=1}^t p((i),S_i) \frac{\psi_{(i)}}{\pot(S_i)}\right) \frac{  \pot(S^*_{t+1})}{\pot(S_0)} 
=
\frac{  \pot(S^*_{t+1})   }{\pot(S_0)} \prod_{i=1}^t \frac{\gamma_{(i)}^{S_i}}{\zeta_{(i)}} .
\end{equation}
Combining~\eqref{eq:kolpa} with the fact $\sum_{\phi \in \mathcal{F}_t(S_0)} p_{\phi} \le 1$, we get
\begin{equation}\label{eq:kialla}
\sum_{\phi \in \mathcal{F}_t(S_0)} \prod_{i=1}^t \frac{\gamma_{(i)}^{S_i}}{\zeta_{(i)}}  \le 
\max_{S \subseteq [m]} \frac{\pot(S_0)}{\Psi(S)}  .
\end{equation}

%\begin{lemma}\label{final_trick}
Let $\zeta = \max_{i \in [m] } \zeta_i$. Then, summing equation~\eqref{eq:kialla} over all possible sets $S_0$ yields   
%\begin{align*}
%\sum_{\phi \in \mathcal{F}_t} \prod_{ i=1 }^t \gamma^{[i]}_{S_i} \le \zeta^t 
%%
%%\sum_{S_0 \subseteq \mathrm{Span}(\theta)} \pot(S_0)\max_{S \subseteq [m]} \frac{1}{\Psi(S)}\\
%\max_{S \subseteq [m]}\sum_{S_0 \subseteq \mathrm{Span}(\theta)}  \frac{\pot(S_0)}{\Psi(S)}
%%
%%\frac{1}{\min_{S \subseteq [m] } \pot(S) }  \sum_{S_0 \subseteq \mathrm{Span}(\theta)} \pot(S_0) 
%\enspace.
%\end{align*}
%%\end{lemma}
%\begin{proof}
\begin{equation} \label{eq:spanos}
\sum_{\phi \in \mathcal{F}_t} 
\prod_{ i=1 }^t \gamma_{(i)}^{S_i} =
\sum_{S_0 \subseteq \mathrm{Span}(\theta)}
\sum_{\phi \in \mathcal{F}_t(S_0)} 
\prod_{ i=1 }^t \gamma_{(i)}^{S_i} \le 
\zeta^t
\sum_{S_0 \subseteq \mathrm{Span}(\theta)}
\sum_{\phi \in \mathcal{F}_t(S_0)} 
\prod_{ i=1 }^t \frac{\gamma_{(i)}^{S_i}}{\zeta_{(i)}} 
\le  
%\frac{\zeta^t }{\max_{S \subseteq [m] } \pot(S) }  \sum_{S_0 \subseteq \mathrm{Span}(\theta)} \pot(S_0)
\max_{S \subseteq [m]}\sum_{S_0 \subseteq \mathrm{Span}(\theta)}  \frac{\pot(S_0)}{\Psi(S)}. 
\end{equation}
%\end{proof}

\begin{proof}[Proofs of Theorem~\ref{soft} and Remark~\ref{cor:faster_back} ]
Combining~\eqref{eq:spanos} with Corollary~\ref{corollary_basic}, we see that the binary logarithm of the probability that the algorithm does not encounter a flawless state within $t$ steps is at most $t  \log_2 \zeta + T_0$, where 
\begin{eqnarray*}
 T_0 		& =&   \log_2 \left( \max_{\sigma \in \Omega}  \frac{ \theta(\sigma) }{\mu(\sigma) } \right) +   \log_2 \left( \sum_{S \subseteq  \mathrm{Span}(\theta)} \pot(S)  \right) + \log_2  \left(  \max_{S \subseteq [m] } \frac{1}{ \pot(S) } \right)   .
\end{eqnarray*}
Therefore, if $t = (T_0 + s) / \log_2 (1/\zeta) \le (T_0 + s) / \delta$, the probability that the algorithm does not reach a flawless state within $t$ steps is at most  $2^{-s}$. This concludes the  proofs of the first part of Remark~\ref{cor:faster_back} and Theorem~\ref{soft}, since $\max_{\sigma \in \Omega} \theta(\sigma) \le 1$ and
\begin{align*}
 \log_2 \left( \sum_{S \subseteq  \mathrm{Span}(\theta)} \pot(S)  \right) + \log_2  \left(  \max_{S \subseteq [m] } \frac{1}{ \pot(S) } \right) \le \log_2  \frac{ \prod_{i =1 }^{m}  \left( 1 + \psi_i \right)}{ (\psi_{\min })^{ m}  }  \le m \log_2 \left( \frac{1 + \psi_{\max} }{ \psi_{\min}} \right) .
\end{align*}

To see the second part of Remark~\ref{cor:faster_back}, let $\mathcal{I}(\theta)$ denote the set comprising the sets of flaw-indices that may be present in a state selected according to $\theta$.  Recall now that when every flaw is \primary, the only equivalence classes of $\mathcal{F}_t$ that contribute to the sum in~\eqref{eq:spanos} are those for which $S_0 \in \mathcal{I}(\theta)$. Thus, for backtracking algorithms, the sum over $S \subseteq  \mathrm{Span}(\theta)$ in the definition of $T_0$ can be restricted to $S \in \mathcal{I}(\theta)$. Finally,  if every flaw is always present in the initial state and $\psi_i \in (0,1]$ for every $i \in [m]$, then  $I(\theta) = \{ F \} $ and $ \log \left( \frac{1}{ \max_{ S \subseteq [m] } \prod_{j \in S } \psi_j  } \right)  = - \log_2 \prod_{j \in [m] } \psi_j  $. This implies that the second and third term in the expression for~$T_0$ in Remark~\ref{cor:faster_back} cancel out, concluding its proof.  

\end{proof}

%
%\blue{We conclude by proving Corollary~\ref{cor:faster_back} concerning pure backtracking
%algorithms, stated in Section~\ref{sec:informal}.  }

\subsection{Other flaw choice strategies}\label{remarks_proofs_strategies}

The only place where we used the fact that the flaw choice is based on a fixed permutation was to assert, in Lemma~\ref{lem:reconstruct}, that the witness sequence of a trajectory determines the sequence of addressed flaws. Thus, our analysis is in fact valid for every flaw choice strategy that shares this property.

One example of such a strategy is  ``pick a random occurring flaw and address it". To implement this, we can fix a priori an infinite sequence of uniformly random permutations $\pi_1, \pi_2, \ldots $ and at the $i$-th step address the lowest indexed flaw present according to $\pi_i$. It is straightforward to see that Lemma~\ref{lem:reconstruct} still holds if we replace $\pi$ with $\pi_i$ therein and in Definition~\ref{def:reconstruct}. 

As a second example, consider the following recursive way to chose which flaw to address at each step (which makes the algorithm non-Markovian). The algorithm now maintains a stack. The flaws present in $\sigma_1$, ordered according to some permutation $\pi$, comprise the initial stack contents. The algorithm starts by addressing the flaw at the top of the stack, i.e., $\pi(\sigma_1)$, as before. Now, though, any flaws introduced in the $i$-th step, i.e., the elements of $B_i$, go on the top of the stack (ordered by $\pi$), while all eradicated flaws are removed from the stack. The algorithm terminates when the stack empties. It is not hard to see that, by taking $S_0$ to be the initial stack contents, popping the flaw at the top of the stack at each step, and adding $S_i$ to the top of the stack (ordered by $\pi$), the sequence of popped flaws is the sequence of addressed flaws.

\section{Graph coloring proofs}\label{proof_improvement_mike}

\subsection{The algorithm} 

To prove Theorem~\ref{improvement_mike} we will generalize the algorithm of Molloy~\cite{molloy2017list} for coloring triangle-free graphs. The main issue we have to address is that in the presence of triangles, the natural generalization of Molloy's algorithm introduces monochromatic edges when the neighborhood of a vertex is recolored. As a result, the existing analysis fails completely even if each vertex participates in just one triangle. To get around this problem, we introduce backtracking steps into the algorithm, whose analysis is enabled by our new convergence condition, Theorem~\ref{soft}.

For each vertex $v \in V$, let $N_v$ denote the neighbors of $v$ and let $E_v = \{\{u_1, u_2\}: u_1,u_2 \in N_v\}$ denote the edges spanned by them. Recall that the color-list of $v$ is denoted by $\mathcal{L}_v$. It will be convenient to treat $\mathrm{Blank}$ also as a color. Indeed, the initial distribution $\theta$ of our algorithm assigns all its probability mass to the state where every vertex is colored $\mathrm{Blank}$. Whenever assigning a color to a vertex creates monochromatic edges, the algorithm will immediately uncolor enough vertices so that no monochromatic edge remains. Edges with two $\mathrm{Blank}$ endpoints are not considered monochromatic. To uncolor a vertex $v$, the algorithm picks a monochromatic edge $e$ incident to $v$ and assigns $e$ to $v$ instead of a color, thus also creating a record of the reason for the uncoloring. Thus, 
\[
\Omega \subseteq \prod_{v \in V}\{\mathcal{L}_v \cup \{\mathrm{Blank}\} \cup E_v \} .
\]
Let $L = (1 + \epsilon)  \frac{ \Delta }{ \ln f } f^{ - \frac{1}{2+2\epsilon} } $ and assume $\Delta$ is sufficiently large so that $L \ge 10$.

\subsubsection{The flaws}

We let $L_v(\sigma) \subseteq (\mathcal{L}_v  \cup \{\mathrm{Blank}\})$  be the set of colors we can assign to $v$ in state $\sigma$ without creating any monochromatic edge. We call these the \emph{available colors for $v$ in  $\sigma$} and note that $\mathrm{Blank}$ is always available. For each $v \in V$, we define a flaw expressing the fact that there are ``too few available colors for $v$,'' namely
\[
B_v  = \left\{  \sigma \in \Omega:  |L_v(\sigma) | < L \right\} .
\]

For each color $c$ other than $\mathrm{Blank}$, let $T_{v,c}(\sigma)$ be the set of $\mathrm{Blank}$ neighbors of $v$ for which $c$ is available in $\sigma$, 
i.e., the vertices that may ``compete" with $v$ for color $c$. For each $v \in V$, we define a flaw expressing the fact that there is ``too much competition for $v$'s available (real) colors,'' namely 
\[
Z_v =  
\left\{  
	\sigma \in \Omega: \sum_{c \in L_v(\sigma) \setminus \mathrm{Blank}
				} 
			| T_{v,c}(\sigma) | 
			> 				
			\frac{L}{10} |L_v(\sigma)| 
\right\}  .
\]

Finally, for each $v\in V$ and $e\in E$ we define a flaw for the fact that $v$ is uncolored (because of $e$), namely
\[
f_{v}^e =  \left\{\sigma \in \Omega: \sigma(v) = e \right\}  .
\]
Let $F_v = B_v \cup Z_v \cup \bigcup_{e \in E} f_v^e$ and let $\Omega^{+} = \Omega -\bigcup_{v \in V} F_v$; thus $\Omega^{+}$ denotes the partial colorings that do not suffer from any of the above flaws.

\begin{lemma}[\cite{molloy2017list}]\label{sufficient_partial_coloring} Given $\sigma \in \Omega^+$, a complete list-coloring of $G$ can be found efficiently.
\end{lemma}

The proof of Lemma~\ref{sufficient_partial_coloring} is a fairly standard application of the (algorithmic) LLL, showing that $\sigma$ can be extended to a complete list-coloring by coloring  all $\mathrm{Blank}$ vertices with actual colors. Thus, the heart of the matter is reaching a state in~$\Omega^{+}$ (i.e., a partial coloring avoiding all the above flaws).

\subsubsection{The flaw choice}

The algorithm can use any $\pi$-strategy in which every $B$-flaw has priority over every $f$-flaw.

\subsubsection{The actions}

To address $f_v^e$ at $\sigma$, i.e., to color $v$, the algorithm simply chooses a color from $L_v(\sigma)$ uniformly at random and assigns it to~$v$. The fact that $B$-flaws have higher priority than $f$-flaws implies that there are always at least $L$ such choices. 

Addressing $B$- and $Z$- flaws is significantly more sophisticated. For each  vertex $v$, for each vertex $u \in N_v$, let $R_u^v(\sigma) \supseteq L_u(\sigma)$ comprise those colors having the property that assigning them to $u$ in state $\sigma$ creates no monochromatic edge except, perhaps, in $E_v$. To address either $B_v$ or $Z_v$ in $\sigma$, the algorithm selects an action according to the following procedure:
\begin{algorithm}
%\caption{Neighborhood Randomizer}
\begin{algorithmic}[1]  % ``1'' puts a number on each line of the algorithm
\Procedure{Recolor}{$v, \sigma$}
\State Assign to each colored vertex $u$  in $N_v$ a uniformly random color from $R_u^v(\sigma)$ \label{refresh:color}
\While {monochromatic edges exist}
	\State Let $u$ be the lowest indexed vertex participating in a monochromatic edge \label{refresh:pick_u}
	\State Let $e$ be the lowest indexed monochromatic edge with $u$ as an endpoint 
	\State Uncolor $u$ by assigning $e$ to $u$ \label{refresh:uncolor}
\EndWhile
\EndProcedure
\end{algorithmic}
\end{algorithm}

%To see why we give  $B$-flaws priority over $Z$-flaws, consider $\sigma \in Z_v$ such that $R_u^v(\sigma) = \{ \mathrm{Blank} \}$ for every $u \in N_v$. Since {\sc Recolor}$(v, \sigma) = \sigma$, addressing $Z_v$ would amount to getting stuck. But since $\sigma \in B_u$ for every $u \in N_v$, addressing these $B$-flaws first creates ``room" to address $Z_v$.

\begin{lemma}\label{bijection}
Let $S'(v,\sigma)$ be the set of colorings that can be reached at the end of Step~\ref{refresh:color} of {\sc recolor($v,\sigma$)} and let $S''(v,\sigma)$ be the set of possible final colorings. Then $|S'(v,\sigma)| = |S''(v,\sigma)|$. 
\end{lemma}
\begin{proof}
Since Steps~\ref{refresh:pick_u}--\ref{refresh:uncolor} are deterministic, $|S''(v,\sigma)| \le |S'(v,\sigma)|$. To prove that $|S''(v,\sigma)| \ge |S'(v,\sigma)|$, we will prove that if $u \in N_v$ has distinct colors in $\sigma_1', \sigma_2' \in S'$, then there exists $z \in V$ such that $\sigma''_1(z) \neq \sigma''_2(z)$. Imagine that in Step~\ref{refresh:uncolor} we also oriented $e$ to point away from~$u$. Then, in the resulting partial orientation, every vertex would have outdegree at most 1 and there would be no directed cycles. Consider the (potentially empty) oriented paths starting at $u$ in $\sigma''_1$ and $\sigma''_2$, and let $z$ be their last common vertex. If $z$ is uncolored, then $\sigma_1''(z) = e_1$ and $\sigma_2''(z) = e_2$, where $e_1 \neq e_2$; if $z$ is colored, then $\sigma_i''(z) = \sigma_i'(u)$.
\end{proof}

\subsection{Proving termination}
Let $D_v$ be the set of vertices at distance 1, 2 or 3 from $v$ and let
%
%
%For a vertex $v$ let $S_v$ denote the set of vertex-indexed flaws that correspond to vertices within the second neighbor of $v$ as well as the  vertex-edge-indexed flaws that correspond to edges in $E_v$. That is,
%\begin{align*}
% S_v  =  \{ B_u, Z_u , f_{w}^e   \text{ s.t. } \mathrm{dist}(u,v)  \le 2 \text{ and }  e \in E_v \text{ and $ e \ni w$ }  \} . 
%\end{align*}
%
\begin{align*}
 S_v  =  \{ B_u  \}_ {u \in D_v} \cup \{ Z_u \}_{u \in D_v} \cup \{ f_{u}^{\{u,w\}} \}_{u,w \in N_v}  . 
\end{align*}

To lighten notation, in the following we write $\gamma^S(f)$ instead of $\gamma^S_f$. Let $q = (1+ \epsilon) \frac{ \Delta}{ \ln \sqrt{f} } \ge 1$.

\begin{lemma}\label{treli_arkouda}  
For every vertex $v \in V $ and edge $e \in E$:
\begin{enumerate}[label=(\alph*)]
\item if $S \not\subseteq S_v$, then $\gamma^S(B_v) = \gamma^S(Z_v) = \gamma^S(f_v^e) = 0$; \label{SNonly}
%\item for each $ f \in \{ B_v, Z_v \}$, $ e = (u_1, u_2) \in E_v$ and set $S \subseteq S_v$ such that  $f_{u_1}^e, f_{u_2}^e \in S$, we have  $\gamma_S(f) = 0$.
\item 
if $S \supseteq \{ f_{u_1}^{\{u_1, u_2\}}, f_{u_2}^{\{u_1, u_2\}}\}$, then 
$\gamma^S(B_v) = \gamma^S(Z_v) = \gamma^S(f_v^e) = 0$;
\label{Oneperedge}
\item $ \max_{S \subseteq F}   \gamma^{S}( f_v^e )  \le \frac{1}{L}  =: \gamma(f_v^e)$; \label{ta:fe}
%\item for each $ f \in \{ B_v, Z_v \}$ and any $ S \subseteq F$ such that $ S \setminus S_v \ne \emptyset$, we have $\gamma_S(f) = 0 $;
\item $ \max_{ S \subseteq F} \gamma^S(B_v)        \le 2 \mathrm{e}^{- \frac{L}{6} }  =: \gamma(B_v)$; \label{ta:b}
\item  $ \max_{ S \subseteq F} \gamma^S(Z_v) \le   3 q \mathrm{e}^{ - \frac{L}{60} } =: \gamma(Z_v)$\label{ta:z},
\end{enumerate}
where the charges are computed with respect to the uniform measure over $\Omega$.
\end{lemma}

We note that, while we give uniform bounds on the charges corresponding to each flaw, the analysis of our algorithm cannot be captured by the result of~\cite{AIK}.  This is  because we will crucially exploit the existence of primary flaws.

Before giving the proof of Lemma~\ref{treli_arkouda}, we first use it to derive Theorem~\ref{improvement_mike}.
\begin{proof}[Proof of Theorem~\ref{improvement_mike}]
For every flaw $f \in F$, we will take $\psi_f = \gamma(f) \psi$, where $\psi >0$ will be chosen later.

For any vertex $v \in V$, flaw $f \in \{B_v, Z_v, f_v^e\}$, and set of flaws $S \subseteq F$, Lemma~\ref{treli_arkouda} implies that $\gamma^S(f) = 0$ unless all $B$- and $Z$-flaws in $S$ correspond to vertices in $D_v$, per part~\ref{SNonly}, and every edge $e \in E_v$ contributes at most one flaw to $S$, per part~\ref{Oneperedge}. Therefore, for $f \in \{B_v, Z_v, f_v^e\}$,
\begin{align}
\frac{ 1}{\psi_f}
\sum_{S \subseteq F } \gamma^S(f) \prod_{g \in S} \psi_g \le \frac{1}{\psi}
\prod_{ u \in D_v } ( 1 + \gamma(B_u) \psi)( 1 + \gamma(Z_u) \psi)  \label{further_bound} 
\prod_{e= \{u_1,u_2\} \in E_v} \left( 1 +   \gamma(f_{u_1}^{e})\psi + \gamma(f_{u_2}^{e}) \psi \right) .
\end{align}

To bound the right hand side of~\eqref{further_bound} we use parts~\ref{ta:fe}--\ref{ta:z} of Lemma~\ref{treli_arkouda} along with the facts $|D_v| \le \Delta^3+1$ and $|E_v| \le \Delta^2/f$ to derive~\eqref{eq:stuct_bound} below. To derive~\eqref{zvbound}, we use the facts that $2\mathrm{e}^{- \frac{L}{6}} \le 3q \mathrm{e}^{-\frac{L}{60}}$, since $q \ge 1$, and that $1+ x \le \mathrm{e}^x$ for all $x$. Thus, for $f \in \{B_v, Z_v, f_v^e\}$, we conclude
\begin{eqnarray}
\frac{1}{\psi_f}
\sum_{S \subseteq F } \gamma^S(f) \prod_{g \in S} \psi_g
& \le & 
\frac{1}{\psi}
	\left( 1 + 2\mathrm{e}^{- \frac{L}{6} } \psi \right)^{\Delta^3+1}  
 	\left( 1 + 3q \mathrm{e}^{- \frac{L}{60} } \psi    \right)^{ \Delta^3+1 } 
	\left(  1+   \frac{2\psi}{L}   \right)^{  \frac {\Delta^2}{ f } }
	\label{eq:stuct_bound} \\
& \le &   
\frac{1}{\psi} \;
\mathrm{exp} \left(    \frac{  2\psi \Delta^2 }{ f L  }  + 6q\mathrm{e}^{-\frac{L}{60}} \psi  (\Delta^3 +1)\right)  =: \frac{1}{\psi} \exp(Q) 
\label{zvbound} .
\end{eqnarray} 

%We now turn to the case of flaws of the form $f_v^e$. At first notice that such a flaw can only introduce flaws of the form $B_u, Z_u$ such that $\mathrm{dist}(v, u) \le 2$. Again, plugging in~\eqref{onepsi} to the conditions of Theorem~\ref{soft} we get that
%\begin{eqnarray}
% \frac{1}{\psi_{f_v^e} } \left(  \sum_{    S \subseteq F } \gamma_S(f_v^e)  \prod_{g \in S} \psi_g     \right)  & \le & \frac{1}{\psi}  \left(  \sum_{S \subseteq F }  \prod_{g \in S} \psi \gamma(g)  \right)   \nonumber \\
% 																				& \le& \frac{1}{\psi} \left(    \prod_{ u \text{ s.t. }   \mathrm{dist}(u,v) \le 2   } ( 1 + \psi   \gamma(B_u) )( 1 + \psi  \gamma(Z_u) )     \right) \nonumber \\
%																				& \le& \frac{1}{\psi} \left( \left( 1 + 2\mathrm{e}^{- \frac{L}{6} } \psi \right)^{\Delta^2+1}  \left( 1 + 3q \mathrm{e}^{- \frac{L}{60} } \psi    \right)^{ \Delta^2+1 } \right) \nonumber  \\
%																				& \le&  \frac{1}{\psi} \left(  \mathrm{exp} \left(    3q\mathrm{e}^{-\frac{L}{60}} \psi  (2\Delta^2 +2)    \right) \right)  \label{final_proof_proper}\enspace,
%\end{eqnarray}
%where we followed similar steps to the previous case to reach~\eqref{final_proof_proper}. Recalling that $\psi = 2 + \epsilon $ and noticing that $ \mathrm{exp} \left(    3q\mathrm{e}^{-\frac{L}{60}} \psi  (2\Delta^2 +2)    \right) $ tends to $1$ as $ \Delta $ grows concludes the proof.

Setting $\psi = (1 + \epsilon)$, we see that the right hand side of~\eqref{zvbound} is strictly less than 1 for all $\Delta \ge \Delta_{\epsilon}$, since $Q\xrightarrow{\Delta \to \infty} 0$  for all $f \in [\Delta^{\frac{2+2\epsilon}{1+2\epsilon} } (\ln \Delta)^2, \Delta^2+1]$. To see this last claim, recall that $L = (1 + \epsilon)  \frac{ \Delta }{ \ln f } f^{ - \frac{1}{2+2\epsilon} }$ and $q = (1+ \epsilon) \frac{ \Delta}{ \ln \sqrt{f} }$, and note that $\ln f < 3 \ln \Delta$ and $f^{ \frac{ 1+ 2 \epsilon}{2 + 2 \epsilon }} \ge \Delta (\ln \Delta)^{ \frac{ 2+ 4 \epsilon}{2 + 2 \epsilon }}$. Thus,
\begin{align}
  \frac{2 \psi \Delta^2 }{f L }  =  
  \frac{2 \Delta^2 }{ f  \frac{ \Delta}{ \ln f } f^{ - \frac{1}{ 2+ 2\epsilon} }  }  = 
  \frac{2 \Delta \ln f  }{f^{ \frac{ 1+ 2 \epsilon}{2 + 2 \epsilon }  }  }  \le
  \frac{2  \ln f  }{( \ln \Delta )^ \frac{2 + 4 \epsilon }{ 2+ 2 \epsilon } }  
  \le 
  \frac{6 \ln \Delta} {  ( \ln \Delta )^ \frac{2 + 4 \epsilon }{ 2+ 2 \epsilon }  } = 
  \frac{6} {  ( \ln \Delta )^ \frac{\epsilon }{ 1 +  \epsilon }  } 
	\xrightarrow{\Delta \to \infty} 0     , \label{trade_off}
\end{align}
while the facts $L = \Omega(\Delta^{\frac{\epsilon}{1+2\epsilon}})$ and $q \le (1+\epsilon) \Delta$ imply that $6q\mathrm{e}^{-\frac{L}{60}} \psi  (\Delta^3 +1) \xrightarrow{\Delta \to \infty} 0$.

\end{proof}

\subsubsection{Proof of Lemma~\ref{treli_arkouda}}

\begin{proof}[Proof of part \ref{SNonly}]
Addressing $B_v$ or $Z_v$ by executing {\sc Recolor}$(v, \cdot)$ only changes the color of vertices in $N_v$, with any resulting uncolorings being due to edges in $E_v$. Thus, only flaws in $S_v$ may be introduced. Addressing $f_v^e$ by coloring $v$ trivially can only introduce flaws $B_u, Z_u$, where $u \in N_v$.
\end{proof}

\begin{proof}[Proof of part \ref{Oneperedge}]
Since addressing an $f$-flaw never introduces another $f$-flaw, we only need to discuss procedure {\sc Recolor}. Therein, vertices are uncolored serially in time, so that any time a vertex $w$ is uncolored there exists, at the time of $w$'s uncoloring, a monochromatic edge $e = \{w,u\}$. Therefore, an edge $e = \{u_1, u_2\}$ can never be the reason for the uncoloring of both its endpoints, i.e., $f^e_{u_1} \cap f^e_{u_2} = \emptyset$.
%
%To see  part (e)  assume that $u_1$ is lower indexed than $u_2$ and that  there exists a state $\sigma$ and an action $\sigma' \in A(f, \sigma)$,  where $f \in \{ B_v, Z_v \}$, such that if  the algorithm takes this action it introduces flaw $f_{u_1}^e$, i.e., it assigns $u_1$ the symbol $(\mathrm{uncolored},e)$. 
%
%Recall that Lemma~\ref{bijection} implies there exists a bijective mapping $\psi $ from $A(f,\sigma)$ to $R(f,\sigma)$. Let $\sigma' = \psi(\sigma')$. Notice now that since the transition $\sigma \rightarrow \sigma'$ introduces flaw $f_{u_1}^e$ it should   be that $e$ is monochromatic in  state $\sigma' \equiv \psi(\sigma')$ right after the algorithm recolored vertices in $ N_v$ by choosing from a color from lists $R_u^v$, $u \in N_v$. Moreover, when the algorithm applies the deterministic procedure that maps $\sigma'$ to $\sigma'$ it considers $u_1$ and assigns it the symbol $(\mathrm{uncolored},e)$ because $e$ has to be the lowest indexed monochromatic edge that contains $u_1$.  At this point, $e$ is no longer monochromatic and, therefore, when the algorithm considers vertex $u_2$ it never assigns it the symbol $(\mathrm{uncolored},e)$. Thus, flaw $f$ cannot introduce both  $f_{u_1}^e$ and $f_{u_2}^e$ simultaneously, concluding the proof.
%
%\blue{Because we uncolor vertices one by one, one of the two endpoints of $e$ will be uncolored first. But an edge only causes the uncoloring of a vertex when both its endpoints are colored (with the same color).}
%
\end{proof}

\begin{proof}[Proof of part~\ref{ta:fe}]
If addressing $f_v^e$ results in $\tau$, then the previous state $\sigma$ must be the mutation of $\tau$ that results from assigning $e$ to $v$. Since $\pi(\sigma) = f_v^e$ implies $\sigma \not\in B_v$, it follows that $|L_v(\sigma)| \ge L$. Since colors are chosen uniformly from $L_v(\sigma)$, it follows that $\gamma(f_v^e) \le 1/L$. 
%\blue{For the purposes of proving parts~(\ref{ta:fe})--(\ref{ta:z}) it will be convenient to think of the matrix $A_S^f$ associated with a flaw $f$ and set $S \subseteq F$ as a directed graph, each non-zero entry $A_S^f(\sigma, \sigma')$ an arc $\sigma \to \sigma'$ of weight $\rho_f(\sigma, \sigma')$. In this view, since $M = I/|\Omega|$ and $\|\cdot \| = \| \cdot\|_1$, the charge $\gamma_S^f = \| M A_S^f M^{-1}\| $ will always be the maximum, over any vertex (state) $\sigma' \in \Omega$, of the sum of the weights of its incoming arcs.}
%The proof of part (a) is relatively straightforward. Recalling that flaws of the form $B_{v'}, v' \in V$, always have priority over flaw $f_v^e$ and Remark~\ref{priority_remark} we get that
%\begin{align*}
%\gamma(f_v^e) \le \max_{\sigma' \in \Omega }  \sum_{ \sigma \in f_v^e \setminus  \bigcup_{v' \in V} B_{v'}   } \rho_{f_v^e}(\sigma, \sigma')  \le   \max_{\sigma' \in \Omega}  \sum_{\sigma \in f_v^e \setminus \bigcup_{v' \in V} B_{v'}   } \frac{1}{L}  \enspace.
%\end{align*}
% This concludes the proof.
%
%\blue{For any set $S \subseteq F$, the charge $\gamma_S(f_v^e)$ is the maximum, over $\sigma' \in \Omega$, of the sum, over the arcs incoming to $\sigma'$ labeled by flaw $f_v^e$, of the probability on each arc. But for every $\sigma'$ there is \red{at most one incoming arc labelled $f_v^e$} and since $B$-flaws have priority, $v$ has at least $L$ available colors, and we pick its color uniformly.}  
\end{proof}

\begin{proof}[Proof of parts~\ref{ta:b} and~\ref{ta:z}]

Observe that every flaw corresponding to an uncolored vertex is \primary, since procedure {\sc Recolor} never colors an uncolored vertex and addressing $f_v^e$ only colors $v$. Thus,  when computing $\gamma^S(f)$, for $f \in \{B_v, Z_v\}$ and $S \subseteq F$, we can restrict to pairs $(\sigma, \tau)$ such that the set of uncolored vertices in $\tau$ is exactly the union of the set of uncolored vertices in $\sigma$ and the set $\{u \in N_v : f_u^e \in S \}$. Fixing $f \in \{B_v, Z_v\}$, $S \subseteq F$, and $\tau$, let us denote by $\mathrm{In}^{S}_f(\tau)$ the candidate set of originating states, and by $\mathcal{U}_f^S(\tau)$ their common set of uncolored vertices. Then, for any $f \in \{B_v, Z_v\}$ and any $S \subseteq F$,
%\marginpar{\tiny FI: Should we change this notation or is it too dangerous?}
\begin{align}\label{charge_in_def}
\gamma^{S}(f)  = \max_{ \tau \in \Omega }  \sum_{ \sigma \in  \mathrm{In}_f^S( \tau)} \rho_{f}(\sigma, \tau) .
\end{align}
To bound $\rho_f(\sigma,\tau)$ in~\eqref{charge_in_def}, we recall that {\sc Recolor} assigns to each colored vertex $u \in N_v$ a random color from $R_u^v(\sigma)$ and invoke Lemma~\ref{bijection} to derive the first equality in~\eqref{almost_there}. For the second equality we observe that for every $u \in N_v$, the set $R_u^v$ is determined by the colors of the vertices in $V \setminus N_v$. Since {\sc Recolor} only changes the color of vertices in $N_v$, it follows that $R_u^v(\sigma) = R_u^v(\tau)$, yielding 
\begin{align}\label{almost_there}
\rho_{f}(\sigma,\tau) = \frac{ 1}{ \prod_{u \in N_v \setminus \mathcal{U}^S_f(\tau)} | R_u^v(\sigma) |  } =   \frac{ 1}{ \prod_{u \in N_v \setminus  \mathcal{U}_f^S(\tau)} | R_u^v(\tau) |  } :=  \frac{1}{ \Lambda_f^S(\tau)} . 
\end{align}

Next we bound $|\mathrm{In}_{f}^S(\tau)|$ as follows. First we observe that if $\sigma \in \mathrm{In}^{S}_f(\tau)$, then $\sigma(u) \neq \tau(u)$ implies $u \in N_v \setminus  \mathcal{U}_f^S(\tau)$ and, therefore, $\sigma(u) \in R_u^v(\tau)$ since $ \sigma(u) \in L_u(\sigma) \subseteq R_u^v(\sigma)  = R_u^v(\tau)$. Thus, the set of $\tau$-mutations that result from recoloring each vertex in $N_v \setminus  \mathcal{U}_f^S(\tau)$ with a color from $R_u^v(\tau)$ so that the resulting state belongs to~$f$ is a superset of $\mathrm{In}_{f}^S(\tau)$. Denoting this last set by $\mathrm{Viol}(f,\tau)$, we conclude that
\begin{align}
\gamma^S(f) = \max_{\tau \in \Omega}  \frac{  |\mathrm{In}_{f}^S(\tau)| }{ \Lambda_f^S(\tau) } \le \max_{\tau \in \Omega} \frac{ |\mathrm{Viol}(f,\tau)  |}{ \Lambda_f^S(\tau) }   = \max_{\tau \in \Omega} \Pr [{\text{\sc recolor}}(v,\tau) \in f]   , \label{eq:bearhands}
\end{align}
where for the last equality we use the definition of  $\Lambda_f^S(\tau)$. 
\begin{remark}
We note that expressing the sum of the transition probabilities into a state in terms of a random experiment, as we do in~\eqref{eq:bearhands}, was the key technical  idea of~\cite{molloy2017list} in order to apply the entropy compression method. It is also the one that breaks down if we allow our algorithm to go through improper colorings. 
\end{remark}

To conclude the proof of Lemma~\ref{treli_arkouda} we prove the following bounds in Appendix~\ref{proof_of_key_bounds2} via fairly routine calculations.
\begin{lemma}\label{key_bounds2} For each vertex $v$ and $\sigma \in \Omega$:
\begin{enumerate}[label=(\alph*)]
\item $\Pr [{\text{\sc recolor}}(v,\sigma) \in B_v] \le 2   \mathrm{e}^{ - \frac{L}{6}}$. \label{eq:L6}
\item $\Pr [{\text{\sc recolor}}(v,\sigma) \in Z_v] \le 3 q \mathrm{e}^{ - \frac{L}{60}}$.
\label{eq:L60}
\end{enumerate}
\end{lemma}

\end{proof}

\section{Applications to backtracking algorithms}\label{SoftcoreLLL}

An important class of algorithms naturally devoid of ``collateral fixes'' are \emph{backtracking} algorithms. In particular, consider a Constraint Satisfaction Problem (CSP) over a set of variables  $V = \{v_1, v_2 \ldots, v_n \}$, each variable $v_i$ taking values in a domain $\mathcal{D}_{i}$, with a set of constraints  $\mathcal{C} = \{c_1, c_2, \ldots,c_m \}$  over these variables. The backtracking algorithms we consider operate as follows. (Note that in Step~\ref{state:init}, we can always take $\theta$ to be the distribution under which all variables are unassigned; this does not affect the convergence condition~\eqref{eq:main_soft_condition} but may have a mild effect on the running time.)

\begin{algorithm}\caption*{{\bf Generic Backtracking Algorithm}}\label{GenericBackAlgo}
\begin{algorithmic}[1]
\State Sample a partial non-violating  assignment $\sigma_0$ according to a distribution $\theta$ and set $i = 0$ \label{state:init}
\While {unassigned variables exist } 

\State Let $v$ be the lowest indexed unassigned variable in $\sigma_i$
\State Choose a new value for $v$ according to a state-dependent  probability distribution
\If {one or more constraints are violated}
\State Remove the values from enough  variables so that no constraint is violated 
\EndIf

\State Let $\sigma_{i+1}$ be the resulting assignment

\State $i \leftarrow i+1$

\EndWhile
\end{algorithmic}
\end{algorithm}

Let $\Omega$ be the set of partial assignments to $V$ that do not violate any constraint in $\mathcal{C}$. For each variable $v_i \in V$, let flaw $f_i \subseteq \Omega$ comprise the partial assignments in which $v_i$ is unassigned. Clearly, each flaw $f_i$ can only be removed by addressing it, as addressing any other flaw can only unassign $v_i$. Thus, every flaw is \primary\ and a flawless state is a complete satisfying assignment.

In this section we present applications of our main theorem to analyze backtracking
search algorithms.  First, we prove a useful corollary of Theorem~\ref{soft} that holds in the
so-called {\it variable setting}. Second, we analyze a backtracking algorithm of Esperet 
and Parreau~\cite{acyclic} for acyclic  edge coloring that lies outside the variable setting.  
In particular, we recover their $4\Delta$ bound on the acyclic chromatic index and, further,
we show how it can make constructive an existential result of Bernshteyn~\cite{bernshteyn2016new}.
We emphasize that all these analyses follow very easily from our framework.

\subsection{The variable setting}

In this section we show how we can use Theorem~\ref{soft} to employ backtracking algorithms in order to capture applications in the variable setting, i.e., the setting considered by Moser and Tardos. In particular, we consider a product measure over variables $V$ and define a bad event for each constraint $c \in \mathcal{C}$ being violated.  We show the following corollary of Theorem~\ref{soft}. 

\begin{theorem}\label{BackProduct}
 Let $P$ be any product measure over a set of variables $V$ and let $A_c$ be the event that constraint $c$ is violated. If there exist positive real numbers $\{ \psi_v \}_{v \in V}$ such that for every variable $v \in V$,
\begin{align}\label{tade_sat}
\frac{1}{\psi_v} \left( 1 + \sum_{c \ni v} P(A_c) \prod_{u \in c} \psi_u \right) < 1
,
\end{align}
then there exists a backtracking algorithm that finds a satisfying assignment after an expected number of $O\left( \log (P_{\min}^{-1} ) + |V| \log_2 \left(  \frac{ 1 + \psi_{\max} }{ \psi_{\min}  } \right)  \right) $steps.
 \end{theorem}

Before proving Theorem~\ref{BackProduct}, we first use it to capture a well-known application of the Lov\'{a}sz Local Lemma to  sparse $k$-SAT formulas  when $P$ is the uniform measure.  For a $k$-SAT formula $\Phi$, we denote its maximum degree by $\Delta \equiv \Delta(\Phi)$, i.e., each variable of $\Phi$ is contained in at most $\Delta$ clauses.  
 
\begin{theorem}\label{simple_sat}
Every $k$-SAT formula $\Phi$   with maximum degree $\Delta < \frac{ 2^{k} }{ \mathrm{e}k} $ is satisfiable. Moreover, there exists a backtracking algorithm that finds a satisfying assignment of $\Phi$ efficiently.
\end{theorem}
\begin{proof}
Setting $\psi_v   = \psi = 2 \alpha > 0 $ we see that  it suffices to find a value $\alpha > 0 $ such that
\begin{align*}
 \frac{ 1}{\psi}  +  \frac{1}{2^k} \Delta  \psi^{k-1}   =   \frac{1}{2 \alpha }   +  \frac{1}{2 } \Delta  \alpha^{k-1} < 1 ,
\end{align*}
which  is feasible  whenever
\begin{align*}
  \Delta   < \max_{\alpha > 0 }  \frac{ 2 \alpha  -1}{ \alpha^k  }  =  \frac{ 2^k }{ k}  \cdot \left( 1  - \frac{1}{k} \right)^{k-1} \le \frac{2^k}{ \mathrm{e}k }  .
\end{align*}
\end{proof}
 
\begin{remark}
In~\cite{2ke} it is shown that using a non-uniform product measure $P$ one can improve the bound of Theorem~\ref{simple_sat} to $\Delta < \frac{ 2^{k+1} }{ \mathrm{e}(k+1)} $ and that this is asymptotically tight. We note that we can achieve the same bound using Theorem~\ref{BackProduct} with the same~$P$, but since this a rather involved LLL application we will not explicitly present it here.
\end{remark} 
 
\subsubsection{Proof of Theorem~\ref{BackProduct}}

We consider the following very simple backtracking algorithm.  Start with each variable unassigned. Then, in each state~$\sigma$, choose the lowest indexed unassigned variable~$v$ and sample a value for it according to the product measure~$P$. If one or more constraints become violated, remove the value from every variable in the lowest indexed violated constraint. 

Let $\Omega$ be the set of partial non-violating assignments. Let  $\mu: \Omega \rightarrow \mathbb{R}$ be the probability measure that assigns to each state $\sigma \in \Omega $  the value  $\mu(\sigma) \propto  \prod_{ \substack{ v \in V \\ v \notin  U(\sigma) } } P( \sigma(v) )$, where for brevity we abuse notation by  letting $P(\sigma(v))$ denote the event that variable $v$ is assigned value $\sigma(v)$. 

 Theorem~\ref{BackProduct} will follow immediately from the following lemma. (For brevity, we will index flaws  with variables instead of integers.)

\begin{lemma}\label{weights_product}
For each vertex  $v$  and set of variables $S \ne \emptyset$,
\begin{align*}
\gamma^S_v=
\begin{cases}
 1    \text{\quad if $S = \emptyset$;} \\
P(A_c) \text{\quad if  $S= c$, where $c$ is a constraint containing $v$;}\\
0 \text{\quad otherwise.} 
\end{cases}
\end{align*}
\end{lemma}
\begin{proof}
Notice that the actions related to flaw $f_v$  can only remove the value from sets of variables that correspond to constraints that contain $v$. Thus, $\gamma^S_{v} = 0 $ for every set $S \ne \emptyset$  that does not correspond to a constraint containing $v$.  
Recalling the definition of charges and $U(\sigma)$, we have
\begin{align}\label{backtracking_prwti}
\gamma^S_v =   \max_{\tau \in \Omega }  \sum_{ \substack{ \sigma \in f_v   \\  S = U(\tau) \setminus \left( U(\sigma) \setminus \{v \}  \right)  } }  \frac{\mu(\sigma)}{ \mu(\tau)} \rho_v(\sigma,\tau).
\end{align}
To see the claim for the case of the empty set, notice that, given a state~$\tau$, there exists at most one state $\sigma$ such that $\rho_{v}(\sigma,\tau) > 0 $ and that $U(\tau) \setminus \left( U(\sigma) \setminus \{v \}  \right) = \emptyset$ . This is because we can uniquely reconstruct $\sigma $ from $\tau$ by removing the value from $v$ at $\tau$.  Then we have
\begin{align*}
\frac{\mu(\sigma) }{ \mu(\tau) } \rho_v(\sigma,\tau)=  \frac{\prod_{ u \in V \setminus U(\sigma) }  P(\sigma(u) )}{ \prod_{ u \in V \setminus U(\tau) }  P(\tau(u) ) }    P(\tau(v) )   = \frac{1}{ P(\tau(v) )  } P(\tau(v) ) = 1 .
\end{align*}
To see the claim for the case where $ S = c$, consider the set  $\mathrm{viol}(c)$ consisting of the set of assignments of the variables of $c$ that violate $c$. Notice now that  for every state $\tau \in \Omega$  there is an injection from the set of states $\sigma$ such that  that $\rho_{v}(\sigma,\tau) > 0$ and  $S = U(\tau) \setminus  \left( U(\sigma) \setminus \{v \}  \right)$ to $\mathrm{viol}(c)$.  This is because $c$ should be violated at  each such state~$\sigma$, and hence each state $\sigma$ should be of the form $ \sigma = \tau_{\alpha}$ for $\alpha \in \mathrm{viol}(c)$, where $\tau_{\alpha}$ is the state induced by $\tau$ when assigning $\alpha$ to the variables of $c$.  Observe further that, for every state of the form $\tau_{\alpha}, \alpha \in \mathrm{viol}(c)$, we have that
\begin{align}\label{deuteri_backtracking}
\frac{ \mu(\tau_{\alpha} ) }  { \mu(\tau)  } \rho_{v}(\tau_{\alpha},\tau) =   \left( \prod_{ u \in c \setminus \{ v \} } P \left( X_u = \tau_{\alpha} (u) \right)   \right) P(  X_v= \tau(v)  )  = P\left(  A_{c}^{\alpha}  \right) ,
\end{align}
where $P(A_{c}^{\alpha}  ) $ is the probability of the event that the variables of $c$ receive assignment $\alpha$.  Combining~\eqref{deuteri_backtracking} with~\eqref{backtracking_prwti} and the fact that  $P(A_c ) = \sum_{ \alpha \in \mathrm{viol}(c) } P(A_{c}^{\alpha} )$ concludes the proof of Lemma~\ref{weights_product}.
\end{proof}

Finally, plugging Lemma~\ref{weights_product} into Theorem~\ref{soft} concludes the proof of Theorem~\ref{BackProduct}. 

%As far as the running time is concerned, Remark~\ref{cor:faster_back}  implies that if 
%\begin{align*}
% T_0 = \log_2 \max_{\sigma \in \Omega } \frac{ 1} { \mu(\sigma) }  + \log_2 \left(\prod_{v \in V} \psi_{v} \right) + \log_2 \left( \max_{ S \subseteq [ m] }  \frac{1}{  \prod_{v \in S } \psi_v } \right)
% \end{align*}
%then the probability that the algorithm makes $\frac{T_0 + s}{ \delta}$ steps is $2^{-s}$, where $\delta = 1 - \max_{v \in V } \zeta_v$.
% 

\subsection{Acyclic edge coloring}

An edge-coloring of a graph is \emph{proper} if all edges incident to each vertex have distinct colors. A proper edge coloring is \emph{acyclic} if it has no bichromatic cycles, i.e., no cycle receives exactly two (alternating) colors. The smallest number of colors for which a graph $G$ has an acyclic edge-coloring  is denoted by $\chi'_a(G)$.

Acyclic Edge Coloring was originally motivated by the work of Coleman et al.~\cite{coleman2,coleman1} on the efficient computation of Hessians and, since then, there has been a series of papers~\cite{alon_lll,acyclic,Haeupler_jacm,CLLL,mike_stoc,Ndreca} that upper bound $\chi'_a(G)$ for graphs with bounded degree.  The current best result was given recently by Giotis {\it et al.}~\cite{Kirousis}, who showed that  $ \chi'_a(G) \le 3.74 \Delta$ in graphs with maximum degree~$\Delta$.

\subsubsection{A simple backtracking algorithm}\label{AECARA}

We  show how one can apply Theorem~\ref{soft} to recover the main application of the framework of~\cite{acyclic} with a much simpler proof.  We chose the result of~\cite{acyclic} because, although it 
is not the sharpest known bound for the acyclic chromatic index, it is a canonical example of how the
entropy compression method is applied to analyze backtracking algorithms and has inspired many other
results in the area.  (Indeed, the authors in~\cite{acyclic} already give several applications of 
their techniques to other problems besides acyclic edge coloring.)

Let $G $ be a graph with $m$ edges $E = \{ e_1, \ldots, e_m \}$ and suppose we have $q$ available colors. 
\begin{definition}\label{4available}
Given a graph $G=(V,E)$ and a  (possibly partial) edge-coloring of $G$, say that color $c$ is \emph{4-forbidden for $e \in E$} if assigning $c$ to $e$ would result in either a violation of proper edge-coloration, or  a bichromatic 4-cycle containing $e$. Say that $c$ is \emph{4-available} if it is not 4-forbidden.
\end{definition}

\begin{lemma}[\cite{acyclic}]\label{lem:2D}
In any proper edge-coloring of $G$, at most $2(\Delta-1)$ colors are 4-forbidden for any $e \in E$.
\end{lemma}
\begin{proof}
The 4-forbidden colors for $e = \{u,v\}$ can be enumerated as: (i)~the colors on edges adjacent to $u$; and (ii)~for each edge $e_v$ adjacent to $v$, either the color of $e_v$ (if no edge with that color is adjacent to $u$), or the color of some edge $e'$ which together with $e, e_v$ and an edge adjacent to $u$ form a cycle of length $4$. 
\end{proof}

Consider the following backtracking algorithm for Acyclic Edge Coloring with $q = 2(\Delta-1) + Q $ colors. At each step, choose the lowest indexed  uncolored edge $e$ and attempt to color it choosing uniformly at random among the $4$-available colors for $e$.   If one or more bichromatic cycles are created, then choose the lowest indexed one of them, say $C =  \{  e_{i_1}, e_{i_2}, \ldots, e_{i_{2 \ell} }  = e \}$, and remove the colors from all its edges except $e_{i_{1}}$  and $e_{i_{2} } $.

The main result of~\cite{acyclic} states that every graph $G$ admits an acyclic edge coloring with  $ q > 4(\Delta -1)$  colors. Moreover, such a coloring can be found in $O \left( |E| |V|  \Delta^2 \ln \Delta \right)$ time with high probability.

We prove the following theorem, which achieves the same bound on~$q$ and 
improves the running time bound when the graph is dense.  More significantly, our
analysis is much simpler.
\begin{theorem}\label{our_aec}
Every graph $G$ admits an acyclic edge coloring with  $ q > 4(\Delta -1)$  colors. Moreover, such a coloring can be found in $O \left(|E| |V|\Delta \right)$ time with high probability.
\end{theorem}

\begin{proof}

Let $\Omega$ be the set of partial acyclic edge colorings of $G$. For each edge $e$, let $f_{e} $ be the subset (flaw) of $\Omega$ that contains the partial acyclic edge colorings of $G$ in which $e$ is  uncolored. We will apply Theorem~\ref{soft} using the $\| \cdot \|_1$ norm and $M =  {\rm diag}(\frac{1}{| \Omega | })$.

We first compute the charges $\gamma_e^{S}$ for each edge $e$ and set of edges $S$. Notice that for $\gamma_e^S$ to be non-zero, it should either be that $S = \emptyset$, or that $S$ contains $e$ and  there exists a cycle $C = \{e_{i_1}, e_{i_2} \} \cup S$ so that, when a recoloring of $ e$ makes $C$ bichromatic, the backtracking step uncolors precisely the edges in $S$. With that in mind, for each edge $e$ and  each set $S$ that contains $e$, let $\mathcal{C}_e(S)$ denote the set of cycles with the latter property.

\begin{lemma}\label{AEC_weights}
For each edge $e$, let
\begin{align*}
\gamma_e^S=
\left\{
	\begin{array}{ll}
		\frac{1}{Q}	   				& \mbox{if } S = \emptyset \\
		\frac{|\mathcal{C}_e(S)| }{Q }   	& \mbox{if }  e \in S  \\
		0 							& \mbox{otherwise} .
	\end{array}
\right.   
\end{align*}
\end{lemma}

\begin{proof}
Notice that
\begin{align*}
\gamma_e^S  =  \max_{ \tau \in \Omega } \sum_{  \substack{ \sigma  \in f_e \\ S = U(\tau)  \setminus \left( U(\sigma)  \setminus \{ e \} \right) }} \rho_e (\sigma, \tau)  \le \max_{ \tau \in \Omega } \sum_{  \substack{ \sigma  \in f_e \\ S = U(\tau)  \setminus \left( U(\sigma)  \setminus \{ e \} \right) }} \frac{1}{Q} ,
\end{align*}
since, according to Lemma~\ref{lem:2D}, $\rho_e(\sigma,\tau) \le \frac{1}{Q}$ for each pair $(\sigma, \tau) \in f_e \times \Omega$.
The proof follows by observing that, for each state $\tau$:
\begin{itemize}
 \item If $S = \emptyset$, then there exists at most one state $\sigma$ such that $\rho_e(\sigma, \tau) > 0$  and $U(\tau) \setminus \left( U(\sigma) \setminus  \{ \mathrm{e} \} \right)  = \emptyset$ (we can reconstruct $\sigma$ from $\tau$ by uncoloring $e$).
\item If $ S \ni e $ and $|S| = 2 \ell -2 $, then there exist at most $| \mathcal{C}_e(S)  |$ states such that $\rho_e(\sigma,\tau)> 0$ and $S = U(\tau) \setminus \left( U(\sigma)  \setminus  \{ \mathrm{e}  \}\right) $. Given  a cycle $ C = S \cup \{ e_{i_1}, e_{i_2} \} $ we reconstruct $\sigma$ from $\tau$ by finding the colors of edges  in $S \setminus \{ e \} $ from $ \tau(e_{i_1}),  \tau(e_{i_2})$, exploiting the facts that the backtracking step corresponds to an uncoloring of a bichromatic cycle,  edge~$e$ is uncolored, and every other edge has the same color as in~$\tau$. 
\item  For all other $S$  there exists no state $\sigma$ such that $\rho( \sigma ,  \tau) > 0$ and  $S =U(\tau) \setminus \left( U(\sigma) \setminus \{ \mathrm{e} \} \right)  $.

\end{itemize}

\end{proof}

Observe that there are at most $(\Delta-1)^{2 \ell-2}$ cycles of length $ 2 \ell$ containing a specific edge $e$. In other words, there exist at most $(\Delta-1)^{ 2 \ell -3}$  sets of edges $S $  of size $2 \ell -2$ that contain $e$ and such that $\gamma_e^S > 0$ and, in addition, note that we always have $| \mathcal{C}_e(S) | \le \Delta-1$.

Thus, if $ Q = c (\Delta-1)$ for some constant $c$, setting $\psi_e =  \alpha \gamma_{\emptyset}^e  =   \frac{\alpha}{ Q} $, where  $\alpha\in (1,c)$ is a constant, Lemma~\ref{AEC_weights} implies
\begin{eqnarray*}
 \frac{ 1 }{ \psi_e } \left(   \sum_{   S \subseteq E }  \gamma_e^{S} \prod_{ e \in S} \psi_j  \right)  & \le & \min_{ \alpha  \in (1,c) } \left( \frac{1}{ \alpha } +  \sum_{i= 3}^{\infty} \left(\frac{ \Delta-1}{Q} \right)^{2i-2} \alpha^{2i-3} \right) \\
 & \le & \min_{\alpha   \in (1,c)} \left(  \frac{1}{\alpha} +  \frac{1}{c}  \sum_{i=3}^{\infty}  \left( \frac{\alpha}{ c} \right)^{2i-3 } \right) 	\\
 &   = & \min_{\alpha   \in (1,c) } \left( \frac{1}{\alpha} + \frac{  \alpha^3}{ c^2 (c^2 - \alpha^2) }    \right)  =  \frac{2}{c} \\
\end{eqnarray*}
for $\alpha^* = c  \left( \frac{ \sqrt{5} -1 }{2} \right)$. Thus, if $c > 2 $ the probability that the algorithm fails to find an acyclic edge coloring within $\frac{T_0 + s}{ \delta}$ steps is $2^{-s}$, where $\delta = 1 - \frac{2}{c}$, and, according to  Remark~\ref{cor:faster_back},
\begin{align*}
T_0 =  \log_2 |  \Omega |   = O( |E|  ).
\end{align*}
The proof is concluded by observing that each step can be performed in time $O(|V| \Delta) $ (the time it takes to find a $2$-colored cycle containing a given edge, if such a cycle exists, in a graph with a proper edge-coloring).
\end{proof}

\subsubsection{An application of the local cut lemma}

Bernshteyn~\cite{bernshteyn2017local} introduced a non-constructive generalized LLL condition, called the ``Local Cut Lemma", with the aim of to drawing connections between the LLL and the entropy compression method. He later applied it 
in~\cite{bernshteyn2016new} to the problem of Acyclic Edge Coloring, giving improved bounds assuming further constraints on the graph besides sparsity. In particular, he proved the following:

\begin{theorem}[\cite{bernshteyn2016new}]\label{first_bern}
Let $G$ be a graph with maximum degree $\Delta$ and let $H$ be a fixed  bipartite graph. If $G$ does not contain $H$ as 
 a subgraph, then there exists an acyclic edge coloring of $G$ using at most $3 (\Delta + o(1))$ colors.
\end{theorem}

We now show how to use our framework to give a constructive proof of Theorem~\ref{first_bern}. This will follow immediately from the following structural lemma in~\cite{bernshteyn2016new}.

\begin{lemma}[\cite{bernshteyn2016new}]\label{absence_lemma}
There exist positive constants $\gamma,\delta$ such that the following holds. Let $G$ be a graph with maximum degree $\Delta$ that does not contain $H$ as a subgraph. Then, for any edge $e \in E(G)$ and for any integer $k \ge 4$, the number of cycles of length $k$ in $G$ that contain $e$
is at most $\gamma \Delta^{k-2-\delta }$.
\end{lemma}

\begin{proof}[Constructive Proof of Theorem~\ref{first_bern}]
Notice that in this case, making almost identical calculations to those in the proof of Theorem~\ref{our_aec}, invoking Lemma~\ref{absence_lemma} to upper bound the number of cycles that contain $e$  and setting $\alpha = \frac{c}{\beta}$, we obtain
\begin{align*}
 \frac{ 1 }{ \psi_e } \left( \sum_{S \subseteq E  }   \gamma_e^{S} \prod_{ h \in S} \psi_h \right)    \le   \min_{\alpha \in (1,c) }  \left(  \frac{1}{\alpha} + \frac{ (\alpha)^3 \gamma \Delta^{-\delta }  }{ c^2 (c^2 - \alpha^2)  }  \right)   = \frac{1}{c}  \min_{ \beta > 1}  \left( \beta + \frac{ \beta \gamma \Delta^{-\delta} } { \beta ( \beta^2-1) }   \right) .
\end{align*}
Thus, as  $\Delta$ grows,  the value of $c$ required for the algorithm to terminate approaches $1$, concluding the proof.
\end{proof}

%\newpage
%
%

%
%\newpage
%
%\input{graphcoloringproofs}
%
%\input{backtrackingproofs}
%
%\input{witnesslemmaproof}

\section{Acknowledgements}
We are grateful to  Paris Syminelakis for various insightful remarks and comments. We also thank David Harris for helpful discussions. 

\bibliographystyle{plain}
\bibliography{kolmo}

\newpage
\appendix
\section{Matrices and norms}\label{background}

Let  $\| \cdot \|$ be any norm over vectors in $\mathbb{R}^n$. The \emph{dual} norm, also over vectors in  $\mathbb{R}^n$,  is defined as   
\begin{align*}
 \| z  \|_{*} = \sup_{ \| x \| = 1} |  z^{\top } x|  .
\end{align*}
For example, the dual norm of $\|  \cdot  \|_{\infty}$ is  $\| \cdot \|_1 $. It can be seen that $\| \cdot \|_{**} = \| \cdot \| $ and that for any  vectors $x, z$,
\begin{align}\label{basic_inequality}
z^{\top} x  = \|x \|   \left(\frac{ z^{\top}x }{ \| x \| }  \right) \le    \| z  \|_{*} \| x \|  .
\end{align}
The corresponding \emph{operator norm}, over  $n \times n$ real matrices, is defined as
%\begin{align*}
%\| A \| \equiv \sup_{ x \ne 0 } \frac{ \|  A x\|}{ \|  x \| } \enspace.  
%\end{align*}
\begin{align*}
\| A \| \equiv \sup_{ \|x\| = 1} %\frac{ 
\|  A x\|
%}{ \|  x \| } 
.  
\end{align*}
For example, if $A$ is a matrix with non-negative entries then $\| A \|_{\infty}$ and $\| A \|_{1 }$ can be  seen to be  the maximum \emph{row}  and \emph{column} sum of $A$, respectively.  
%The \emph{dual} norm, denoted by $\| \cdot  \|_{*}$ , is defined as 
%\begin{align*}
% \| z  \|_{*} = \sup \left\{  |  z^{\top } x|   \text{ s.t. } \| x \| = 1  \right\} \enspace.
%\end{align*}
Operator norms are submultiplicative, i.e., for every operator  norm $\|  \cdot \|$ and any two $n \times n$  matrices $A, B$,
\begin{align}\label{submult_2}
\| A B \| \le \| A \| \| B \| .
\end{align}
Finally, for any vector norm $\| \cdot \|$, any row vector $x^{\top}$ and any $n \times n$ matrix $A$, we have
\begin{align}\label{submult_1}
\| x^{\top}  A \|_*  \le \| x^{\top} \|_* \| A \| .
\end{align}

\section{The matrix norms framework in action}\label{app:spectral}

%We start by showing how the framework of matrix norms captures the classical \emph{potential function argument}.
In this appendix we illustrate how the framework of matrix norms captures a wide variety of convergence arguments for local search algorithms, both LLL-inspired ones and others. Recall from our discussion in the introduction that our goal is to bound $\rho(A)$, the spectral radius of the transition matrix between flawed states.

We begin with the classical potential function argument.
Consider any function $\phi$  on $\Omega$ such that $\phi(\sigma) > 0$ for $ \sigma \in \Omega^*$, while $\phi(\sigma) =0 $ for $\sigma \notin \Omega^*$. 
%In our $k$-SAT example, $\phi(\sigma)$ could be the number of violated clauses under $\sigma$. 
The potential argument asserts that eventually $\phi =0$ (i.e., the particle escapes $\Omega^*$) if $\phi$ is always reduced in expectation, i.e., if for every  $\sigma \in \Omega^*$,
\begin{align}\label{eq:potential}
\sum_{ \tau \in \Omega} P[\sigma,\tau] \phi(\tau) < \phi(\sigma) .
\end{align}
To express this argument via matrix norms, let $\Amat' = M \Amat M^{-1}$ where $M$ is the diagonal $|\Omega^*| \times |\Omega^*|$ matrix ${\rm diag}(1/\phi(\sigma))$. Thus, $\Amat'[\sigma,\tau] = \Amat[\sigma,\tau] \phi(\tau) / \phi(\sigma)$. Recalling that $\| \cdot \|_{\infty}$ is the maximum row sum of a matrix, we see that condition~\eqref{eq:potential} for the potential function is nothing other than $\|\Amat' \|_{\infty} < 1$, implying $\rho(A) = \rho(A') \le\|\Amat' \|_{\infty} <1$. \smallskip
%
%Our starting point is the observation that all entropy compression arguments, and indeed all arguments in the algorithmic LLL literature, can be seen as \emph{dual} to the potential function argument. That is, after a suitable change of basis $A' = MAM^{-1}$, they bound not $\|A'\|_{\infty}$, as the potential argument, but the dual norm $\|A'\|_1$. 

Next we show how the same approach, using the dual matrix norm $\|\cdot\|_{1}$, captures the Moser-Tardos algorithm for $k$-SAT. Given a $k$-SAT formula on $n$ variables with clauses $c_i$, let $\Omega=\{0,1\}^n$ denote the set of all assignments, and $\Omega^*$ the set of non-satisfying assignments. To simplify exposition, we assume that in each step, the algorithm picks the lowest-indexed unsatisfied clause~$c_i$ and resamples all variables in~$c_i$. Thus, the state-evolution is a Markov chain, and if we denote its transition matrix by $P$, we are interested in the submatrix $A$ that is the projection of~$P$ onto $\Omega^*$.
For each clause $c_i$, let $\Amat_i$ be the $|\Omega^*| \times |\Omega^*|$ submatrix of $A$ comprising all rows (states) where the resampled clause is $c_i$. (All other rows of $A_i$ are~0.)
For $t \geq 1$, let $\mathcal{W}_t$ contain every $t$-sequence of (indices of) clauses that has non-zero probability of being the first $t$ clauses resampled by the algorithm. In other words, $\mathcal{W}_t$ is the set of all \mbox{$t$-sequences} of indices from $[m]$ corresponding to non-vanishing $t$-products of matrices from $\{A_1, \ldots, A_m\}$, i.e., $\mathcal{W}_t = \{W=(w_i) \in [m]^t : \prod_{i=1}^t A_{w_i} \neq 0\}$. For every operator norm $\|\cdot\|$ we get:
\begin{align}\label{eq:lllnorm}
	\rho(\Amat)^t = \rho(A^t) \le \left\|A^t\right\|
	= \biggl\|  \biggl( \sum_{i \in [m] } \Amat_i \biggr)^t \biggr\|
	= \biggl\|\sum_{ W \in \mathcal{W}_t } \prod_{i =1}^t  \Amat_{w_i} \biggr\|
	\le \sum_{ W \in \mathcal{W}_t } \biggl\|\prod_{i =1}^t  \Amat_{w_i}   \biggr\|
	\le \sum_{ W \in \mathcal{W}_t } \prod_{i =1}^t  \left\|\Amat_{w_i}   \right\|.
\end{align}
The first inequality here follows from the fact that $\rho(A) \le \|A\|$ for any operator norm $\|\cdot\|$, the second inequality is the triangle inequality, and the third follows by submultiplicativity of operator norms.

To get a favorable bound, we will apply~\eqref{eq:lllnorm}  with the norm $\|\cdot\|_1$, i.e., the maximum column sum. We see that for all $j \in [m]$, every column of $A_j$ has at most one non-zero entry, since $A_j(\sigma,\tau)>0$ only if $\sigma$ is the (unique) mutation of $\tau$ so that $c_j$ is violated. Recalling that all non-zero entries of $A$ equal $2^{-k}$, we conclude that $\|A_j\|_1 = 2^{-k}$ for all $j \in [m]$. Therefore, $\|\Amat^t\|_1\le|\mathcal{W}_t| 2^{-kt}$. To bound $|\mathcal{W}_t |$ we use a simple necessary condition for membership in $\mathcal{W}_t$ which, by a standard counting argument, implies that if each clause shares variables with at most $\Delta$ other clauses then $|\mathcal{W}_t| \le 2^m (\mathrm{e} \Delta)^t$. Therefore $\rho(A)^t \le 2^m (\mathrm{e} \Delta 2^{-k})^t$, implying that if $\Delta  < 2^k/ \mathrm{e}$ then $1 > \|\Amat\|_1 \ge \rho(\Amat)$ and the algorithm terminates within $O(m)$ steps with high probability.  This matches exactly the Moser-Tardos condition (which is tight).

A very similar argument can be used to capture even the most general existing versions of the algorithmic
LLL~\cite{AIJACM,HV,AIK}, which are described by 
%a general set of 
arbitrary flaws and, for each flaw~$f_i$, an arbitrary corresponding transition matrix~$\Amat_i$ for addressing the flaw.
Note that \eqref{eq:lllnorm} is in essence a weighted counting
of witness sequences, %where 
the weight of each sequence %is 
being the product of the norms
$\| \Amat_{w_i}\|_{1}$. Observe also that in our $k$-SAT example above, the only probabilistic notion was the transition matrix $A$ and we did not make any implicit or explicit reference to a probability measure $\mu$. To cast general algorithmic LLL arguments in this same form, any measure $\mu$ is incorporated as a \emph{change of basis} for the transition matrix~$\Amat$, i.e., we bound $\|A'\|_1 = \|M \Amat M^{-1}\|_1$ as $\sum_{ W \in \mathcal{W}_t } \prod_{i =1}^t \| M \Amat_{w_i} M^{-1} \|_1$, where $M$ is the diagonal $|\Omega^*| \times |\Omega^*|$ matrix ${\rm diag}(\mu(\sigma))$,  similarly to the potential function argument. We thus see that the measure $\mu$  is nothing other than a tool for analyzing the progress of the algorithm. 
%The key bound on the failure probability from~\eqref{eq:lllnorm} then becomes
%$\sum_{ W \in \mathcal{W}_t } \prod_{i =1}^t \| M \Amat_{w_i} M^{-1} \|_1$, which can again be
%effectively bounded under the desired LLL condition.

\section{Comparison with previous LLL conditions}\label{comparison}

Here we give background on the LLL and explain how our main theorem compares with previous algorithmic LLL conditions. 

\subsection{Non-constructive conditions}

The original statement of the LLL asserts, roughly, that, given a family of ``bad" events in a probability space, if each bad event individually is not very likely and, in addition, is independent of all but a small number of other bad events, then the probability 
of avoiding all bad events is strictly positive.  Erd\H{os} and Spencer~\cite{LopsTrav} later noted that independence in the LLL can be replaced by positive correlation, yielding the original version of what is known as  the \emph{Lopsided} LLL,  more sophisticated  versions of which have also been established in~\cite{albert1995multicoloured,dudek2012rainbow}. Below we state the Lopsided LLL in  its most powerful form that holds for arbitrary probability spaces and families of bad events (see e.g.,~\cite[p.228]{mike_book}). 
\newpage
%\marginpar{\tiny AS: Not clear how this second sentence relates to the rest of this paragraph or theorem}
%We remark that the directed graph on $[m]$ where each vertex (event) $i$, below, points to the vertices in $L(i)$ is known as the \emph{lopsidependency} graph and the Lopsided LLL of Erd\H{os} and Spencer~\cite{LopsTrav} corresponds to the special case $b_i = \mu(B_i)$. 
\begin{gLLL}\label{generalLOPS}
Let $(\Omega,\mu)$ be a probability space and let $\mathcal{B} = \{B_1, B_2,\ldots,B_m\}$ be a set of $m$ (bad) events. For each $i \in [m]$, let $L(i) \subseteq	 [m] \setminus \{i\}$ be such that $\mu(B_i \mid \bigcap_{j \in S} \overline{B_j}) \le  b_i$ for every $S \subseteq [m] \setminus (L(i) \cup \{i\})$. If there exist positive real numbers $\{\psi_i\}_{i=1}^m$ such that, for all $i \in [m]$,
\begin{equation}\label{eq:LLL}
\frac{b_i}{\psi_i}  \sum_{ S \subseteq   L(i) \cup \{i \} }  \prod_{j \in S} \psi_j \le 1  ,
\end{equation}
then the probability that none of the events in $\mathcal{B}$ occurs is at least $\prod_{i=1}^m 1/(1+\psi_i)> 0$. 
\end{gLLL}

Writing $x_i = \psi_i / (1+\psi_i) \in [0,1)$, condition~\eqref{eq:LLL} takes the more familiar form $b_i \le x_i \prod_{j \in L(i)} (1- x_j)$. The form in~\eqref{eq:LLL}, though, is more amenable to refinement (under additional assumptions)  and comparison. 

The above form of the LLL is motivated by the fact that, in complex applications, small but non-vanishing correlations tend to travel arbitrarily far in the space~$\Omega$.  To isolate these dependencies so that they can be treated locally,  it can be crucial~\cite{albert1995multicoloured,dudek2012rainbow,kahnChrom,kahnListChrom} to allow mild negative correlations between each bad event $B_i$ and the events outside its ``special" set $L(i)$, achieved by allowing $b_i \ge \mu(B_i)$.  

\subsubsection{Extensions}\label{non_constructive_extensions}

In the setting of the General LLL above, we often define the \emph{lopsidependency} graph, which is  the directed graph on $[m]$ in which each vertex $i$ points to the vertices in~$L(i)$. 
In many applications the lopsidependency graph is in fact symmetric, and even when it is not
it is useful to consider an undirected version in which the neighborhood of each vertex $i\in[m]$
includes~$L(i)$ and possibly some other vertices.  
(One can trivially achieve this by simply ignoring the direction of the edges
in the lopsidependency graph.)  Given such an undirected graph~$G$, one can restrict
the sum in~\eqref{eq:LLL} to $S \subseteq L(i)$ that are {\it independent\/} in~$G$. This yields the so-called \emph{cluster expansion} condition~\cite{bissacot}. 
It should be clear that, if the lopsidependency graph is already symmetric (so that the neighborhood
of~$i$ is exactly~$L(i)$), then this leads to
an improvement over~\eqref{eq:LLL}.

In fact, for a given graph $G$ as above, the \emph{exact} condition for avoiding all bad events as a function of the bounds $\{b_i\}_{i=1}^m$ is known, and is due to Shearer~\cite{Shearer}. Unlike the cluster-expansion condition, though, Shearer's condition involves a separate condition for every independent set in $G$. Moreover, when $\mu$ is a product measure and the dependencies between events can be expressed in terms of variable sharing (that is, under the assumptions of the variable setting~\cite{MT}), several works~\cite{CompHarr,he2017variable,szege_meet} have shown that Shearer's condition can be improved, i.e., that more permissive conditions exist. 

%  Shearer's condition was made constructive in the variable setting by Kolipaka and Szegedy~\cite{szege_meet} and in general probability spaces by Harvey and Vondr\'{a}k~\cite{HV} and Kolmogorov~\cite{kolmofocs}. An important notion introduced in~\cite{szege_meet}  is the so-called \emph{stable set matrix} of an  LLL instance, which has an entry for each pair of independent sets in $G$. The importance of the stable set matrix comes from the fact that Shearer's condition is equivalent to the spectral radius of this matrix being strictly less than $1$. We emphasize that this spectral view of Shearer's condition is entirely different from the techniques we introduce in this paper, wherein randomized algorithms are analyzed by bounding the spectral radius of a submatrix of their transition matrix.

Finally, Scott and Sokal~\cite{scott2005repulsive} introduced a so-called \emph{soft-core} LLL condition, in an effort to quantify interaction strengths between bad events (whereas in all other works two bad events either interact or they don't). Unlike our present work, which quantifies general point-to-set interactions, the condition in~\cite{scott2005repulsive} only quantifies \emph{pairwise} (i.e., point-to-point) interactions. Finding combinatorial applications for the soft-core LLL condition was left as an open question in~\cite{scott2005repulsive}. To the best of our knowledge, it remains open.  

\subsection{Constructive conditions}

Using the framework of Section~\ref{sec:informal} for local search algorithms, we say that flaw $f_i$ \emph{causes} flaw $f_j$, and write $i \to j$, if there exist $\sigma \in f_i, \tau \in f_j$ such that $\rho_i(\sigma, \tau)>0$ and the transition $\sigma \to \tau$ introduces flaw $f_j$. (Thus, causality is to be interpreted as \emph{potential} causality.) Let $\Gamma(i) = \{j : i \to j\}$ be the set of flaws caused by $f_i$. We  call the digraph over $[m]$ in which each vertex $i$ points to the vertices in $\Gamma(i)$ \mbox{the \emph{causality} digraph.}

Let $\mu$ be an arbitrary probability measure on~$\Omega$. For $i \in [m]$ and $\tau \in \Omega$, let $\nu_i(\tau)$ be the probability of ending up in state $\tau$ after sampling a state $\sigma \in f_i$ according to $\mu$, and then addressing $f_i$ at $\sigma$. The \emph{distortion} associated with $f_i$ is defined as
\begin{align*}
d_i :=  \max_{\tau \in \Omega } \frac{\nu_i (\tau) }{\mu (\tau) }  \ge 1 .
\end{align*}
Thus, the distortion of  $\mu$ associated with flaw $f_i$ is the greatest inflation of a state probability induced by sampling $\sigma \in f_i$ according to $\mu$ and addressing flaw $f_i$ at $\sigma$. As shown in~\cite{IS}, if we set $b_i = \mu(f_i) d_i$ for all $i \in [m]$, then a causality digraph is a lopsidependency graph for $\mu$, thus making the following algorithmic LLL condition, established in~\cite{AIK}, the algorithmic counterpart of the General LLL~\eqref{eq:LLL}: 

\begin{aLLL}
Let $\gamma_i := \mu(f_i) d_i$. If there exist positive real numbers $\{\psi_i\}_{i =1}^m$ such that, for all $i \in [m] $,
\begin{align} \label{eq:AlgoLLL}
 \frac{\gamma_i}{\psi_i} \sum_{ S \subseteq \Gamma(i) } \prod_{j \in S } \psi_j  < 1    ,
\end{align}
then a local search algorithm as above reaches a flawless object quickly with high probability.
\end{aLLL}

If $d_i = 1$, i.e., $\nu_i(\tau) = \mu(\tau)$ for all $\tau \in \Omega$, we say that the algorithm is a \emph{resampling oracle}~\cite{HV} for $f_i$. When one has resampling oracles for all flaws, i.e., $d_i = 1$ for all $i \in [m]$, then condition~\eqref{eq:LLL} is the algorithmic counterpart of the Lopsided LLL, as established by Harvey and Vondr\'ak~\cite{HV}. Condition~\eqref{eq:AlgoLLL} also subsumes the flaws/actions condition of~\cite{AIJACM}: in that setting $\rho_i(\sigma,\cdot)$ is uniform over the set of possible next states, while the analysis does not reference a measure $\mu$. Taking $\mu$ to be uniform and applying condition~\eqref{eq:AlgoLLL} in fact sharpens the condition of~\cite{AIJACM}.

Observe that a resampling oracle perfectly removes the conditioning on the old state $\sigma$ belonging to $f_i$, since the new state $\tau$ is distributed according to $\mu$.  For example, if $\mu$ is a product measure, this is precisely what is achieved by the resampling algorithm of Moser and Tardos~\cite{MT}. Requiring an algorithm to be a resampling oracle for every flaw may be impossible to achieve by local exploration within $\Omega$, i.e., by ``local search."  (Note that restricting to local search is crucial since longer-range resampling, even if it were possible, would tend to rapidly densify the causality digraph.)  Allowing distortion greater than~1 frees the algorithm from the strict resampling requirement of perfect deconditioning.
Optimizing the tradeoff between distortion and the density of the causality digraph has recently led to strong algorithmic results~\cite{AIK,LLLWTL,IS, molloy2017list}.

\subsubsection{Our new condition}

Our Theorem~\ref{thm:softprelim} is a strict generalization of the above algorithmic LLL condition~(\ref{eq:AlgoLLL}). To see this, observe that $\gamma_i^S \le \gamma_i^{\emptyset} = \gamma_i$ for every set $S \subseteq [m]$, as
\begin{align*}
\gamma_i^S  \le  \max_{\tau \in \Omega} \left\{\frac{1}{\mu(\tau)}\sum_{\sigma \in f_i} \mu(\sigma) \rho_i(\sigma,\tau) \right\}
= \max_{\tau \in \Omega} \frac{\nu(\tau)}{\mu(\tau)}  \mu(f_i)
= d_i \, \mu(f_i) = \gamma_i.
\end{align*}
Hence, since $\gamma_i^S = 0$ for $S \not\subseteq \Gamma(i)$, the l.h.s.\ of our condition~(\ref{eq:softprelim}) is never larger than the l.h.s.\ of~(\ref{eq:AlgoLLL}).

\subsubsection{Extensions}

Condition~\eqref{eq:AlgoLLL} can be improved under additional assumptions. Specifically, let $G$ be an \emph{undirected} graph on $[m]$ such that $\Gamma(i)$ is a subset of the neighbors of $i$, for every $i \in [m]$. (As discussed in Section~\ref{non_constructive_extensions}, one can again trivially get such a $G$ by ignoring the direction of edges in the causality graph.) In~\cite{AIK} it was shown that condition~\eqref{eq:AlgoLLL} can be replaced by the cluster expansion condition~\cite{bissacot} on $G$, while in~\cite{kolmofocs} it was shown that condition~\eqref{eq:AlgoLLL} can be replaced by Shearer's condition~\cite{Shearer}. Both of these conditions benefit by restricting consideration to independent sets of $G$, as discussed in Section~\ref{non_constructive_extensions}. Finally, Kolmogorov~\cite{kolmo_new} showed a new condition that interpolates between the cluster expansion condition and our new condition.

Harris and Srinivasan~\cite{LLLLBeyond,PartResmp1} have developed improved conditions for the convergence of algorithms operating in the so-called variable setting~\cite{MT}, based on refinements of the notion of dependency between bad events. These improvements are incomparable to condition~\eqref{eq:AlgoLLL}, as they do not apply to general local search algorithms (for instance, all algorithms in the variable setting are commutative~\cite{kolmofocs}).

\section{Proof of Lemma~\ref{key_bounds2}}\label{proof_of_key_bounds2}

Our computations are similar to the ones in~\cite{molloy2017list}. The following version of Chernoff Bounds will be useful:
\begin{lemma}\label{ChernoffNegative}
Suppose $\{X_i\}_{i=1}^m \in \{0,1\}$ are boolean variables, and set $Y_i = 1-X_i $, $X = \sum_{i=1}^{m} X_i$.  If $\{ Y_i \}_{i=1}^m $ are negatively correlated, then for any $ 0 < t \le \ex [X ]$
\begin{align*}
\Pr[  | X - \ex[ X] | > t  ] < 2 \exp\left(-\frac{t^2}{3 \ex [ X ]}\right) .
\end{align*}
\end{lemma}
\begin{proof}[Proof of part~\ref{eq:L6} of Lemma~\ref{key_bounds2}] 
Let $v \in V$ and $\sigma \in \Omega$ be arbitrary and let $\tau \in \Omega$ be the (random) output (state) of {\sc Recolor}$(v,\sigma)$. For each color $c \in \mathcal{L}_v$, let $P_v^c = \{u \in N_v : c \in R_u^v(\sigma)\}$ and define
\begin{align*}
 \rho(c)   =   \sum_{u \in P_v^c}  \frac{ 1 }{    |R_u^v(\sigma)|     - 1 }  . 
\end{align*}
Since $c \in R_u^v(\sigma)$ implies $|R_u^v(\sigma)| \ge 2$, and since $1-1/x > \exp(-1/(x-1))$ for $x \geq 2$, we see that
\begin{equation}
\ex[| L_v(\tau)  |] =  1+  \sum_{ c \in \mathcal{L}_v  } \prod_{ u \in P_v^c } \left( 
1 -  \frac{ 1}{   |R_u^v(\sigma)|  }  \right)
>  \sum_{ c \in \mathcal{L}_v }  \prod_{ u \in P_v^c  } \exp\left(-\frac{1}{|R_u^v(\sigma)| - 1}\right)
 =   \sum_{ c \in \mathcal{L}_v  } 	\mathrm{e}^{ -  \rho(c)  }							   \label{humble}  .
\end{equation}
Also, since each $R_u^v(\sigma)$ has $|R_u^v(\sigma)|- 1$ non-Blank colors, we see that
\begin{align}
Z_v := \sum_{c \in \mathcal{L}_v} \rho(c)  
\le  \sum_{u \in N_v}  \sum_{ c \in R_u^v(\sigma) \setminus \mathrm{Blank}  }    \frac{1 }{  |R_u^v(\sigma)| - 1}   
\le  \Delta  \label{humble2}  .
 \end{align}

The fact that $e^{-x}$ is convex implies that the right hand side of~\eqref{humble} is at least  $|\mathcal{L}_v| \exp(-Z_v/|\mathcal{L}_v|)$. Recalling that $|\mathcal{L}_v| = q = (1+ \epsilon) \frac{ \Delta}{ \ln \sqrt{f} }$ and combining~\eqref{humble} with~\eqref{humble2} yields
\[
\ex [ | L_v(\tau)   | ]  
> q \mathrm{e}^{ - Z_v/q } 
\ge (1 + \epsilon)  \frac{ \Delta }{  \ln \sqrt{f}}   \mathrm{e}^{- \Delta/q} 
= 2(1 + \epsilon)  \frac{  \Delta }{  \ln f}   f^{ - \frac{1}{2(1+\epsilon)}}
		 =2L . 
\]
Let $X_c$ be the indicator variable that $c \in L_v(\tau)$ so that $|L_v(\tau)| = 1+ \sum_{ c \in \mathcal{L}_v(\tau)} X_c$. It is not hard to see that the variables $Y_c = 1 - X_c$ are negatively correlated, so that applying Lemma~\ref{ChernoffNegative} with $ t = \frac{1}{2} \ex [ |L_v(\tau)| ]> L $ yields
\begin{align*}
\Pr \left[ |L_v(\tau)|  <  \tfrac{1}{2} \ex \left[|L_v(\tau)|\right] \right]  \le 2 \mathrm{e}^{ - \ex \left[ |L_v(\tau)|  \right] /12 } < 2 \mathrm{e}^{ - L/6 }.
\end{align*}
This concludes the proof.
\end{proof}

\begin{proof}[Proof of part~\ref{eq:L60} of Lemma~\ref{key_bounds2}]
Let $\Psi = \{c \in L_v(\sigma) : \rho(c)  \ge L/20\} \setminus \mathrm{Blank} $. The probability that $L_v(\tau)$ contains at least one color from $\Psi$ is at most
\begin{equation*}\label{eq:lalakis}
\ex \left[ | L_v(\tau) \cap \Psi  | \right] = \sum_{ c \in \Psi}  \prod_{ u \in P_v^c  } \left( 1 -  \frac{1}{   |R_u^v(\sigma)| }  \right) 
<
\sum_{ c \in \Psi }  \prod_{ u \in P_v^c  } \exp\left(-\frac{1}{2(|R_u^v(\sigma)| - 1)}\right)
<  \sum_{ c \in \Psi} \mathrm{e}^{ - \rho(c)/2  }   ,
\end{equation*}
where we used that $c \in R_u^v(\sigma)$ implies $|R_u^v(\sigma)| \ge 2$, and that $1-1/x < \exp(-1/(2(x-1)))$ for $x \geq 2$. Finally note that 
$\sum_{ c \in \Psi} \mathrm{e}^{ - \rho(c)/2  } \le q \mathrm{e}^{-L/40}$ by the definition of the set $\Psi$.%, the right hand side of~\eqref{eq:lalakis} is bounded by $q \mathrm{e}^{-L/40}$.

Recall that $T_{v,c}(\tau) = \{ u \in N_v : \tau(u) = \mathrm{Blank} \text{ and $c \in L_u(\tau)$}\}$. Since $L_u(\tau) \subseteq R_u(\tau) = R_u(\sigma)$, it follows that $T_{v,c} (\tau) \subseteq P_v^c$ and, therefore, $\ex \left[  |T_{v,c} (\tau)| \right] \le \sum_{ u \in P_v^c} 1/|R_u^v(\sigma)| \le  \rho(c)$. Since the vertices in $P_v^c$ are colored (and thus become $\mathrm{Blank}$) independently and since $\rho(c)<L/20$ for $c \not\in \Psi$, applying Lemma~\ref{ChernoffNegative} with $t =L/20$  yields
$\Pr \left[ |T_{v,c}(\tau) | > \ex \left[ | T_{v,c}(\tau) | \right] + L/20 \right] < 2\mathrm{e}^{- L/60}$. Applying the union bound over all $q$ colors, we see that the probability there is at least one $c \notin \Psi$ for which $|T_{v,c}(\tau)| > L/10$ is at most $2q \mathrm{e}^{-L/60}$.  Thus, with probability at least $1 - 3 q \mathrm{e}^{ - L/60}$, 
\begin{align*}
\sum_{c \in L_v(\tau) \setminus \mathrm{Blank} }  |T_{v,c} (\tau) |= \sum_{ c \in L_v(\tau) \setminus  ( \Psi  \cup \mathrm{Blank}  )  } | T_{v,c}(\tau) | < \frac{L}{10} |L_v(\tau)| .
\end{align*}
\end{proof}

\section{Proof of Theorem~\ref{sparse_graphs}}\label{determ_graphs_proof}

%The second part of Theorem~\ref{sparse_graphs}, concerning the case $f \ge \Delta^{\frac{2 +2 \epsilon}{1 +2 \epsilon }  } (\ln \Delta )^2$, follows immediately from Theorem~\ref{improvement_mike}, since $\chi(G) \le \chi_{\ell}(G)$, for every graph $G$. 
%To prove the first part, i.e., for arbitrary $f \ge f_{\epsilon}$, 

We will follow closely the approach adopted by the authors in~\cite{alon1999coloring}.  Throughout the proof we assume that $\epsilon  \in (0,\epsilon_0)$, where $\epsilon_0$ is sufficiently small, and that $f_{\epsilon} > 0 $ and $\Delta_{\epsilon} > 0$ are sufficiently large.

We distinguish two cases, depending on whether $f \ge \Delta^{ ( 2+ \epsilon^2) \epsilon }$ or not. To prove Theorem~\ref{sparse_graphs} for the case $f \ge \Delta^{ ( 2+ \epsilon^2) \epsilon }$ we will prove the following. 

\newcommand\fdelta{\zeta}
\newcommand\fepsilon{\theta}

\begin{theorem}\label{sparse_case}
For every $\fepsilon , \fdelta  \in (0,1)$, there exists    $ \Delta_{\fepsilon,\fdelta} > 0 $  such that  every graph  $G$ with maximum degree $\Delta \ge \Delta_{\fepsilon,\fdelta} $ in which the neighbors of every vertex span at most $\Delta^{ 2 - (2  + \fdelta) \fepsilon }$ edges, has chromatic number $\chi(G) \le   ( 1 +  \fdelta) (1+ \fepsilon^{-1} ) \frac{\Delta}{ \ln \Delta }$. 
\end{theorem}
\begin{proof}[Proof of Theorem~\ref{sparse_graphs} for $f \ge \Delta^{( 2+ \epsilon^2) \epsilon } $]
We apply Theorem~\ref{sparse_case} with $\zeta = \epsilon^2$ and $\theta = \frac{ \ln f}{ (2+ \epsilon^2) \ln \Delta} \ge \epsilon$, so that $\Delta^2 / f = \Delta^{ 2 - (2  + \fdelta) \fepsilon }$. Since $\zeta, \theta <1$, we obtain
\begin{eqnarray*}
\chi(G) &\le &  ( 1 + \fdelta) \left( 1 + \frac{(2 + \fdelta) \ln \Delta }{\ln f } \right) \frac{ \Delta }{  \ln  \Delta }  \\
            & = &   ( 1 + \fdelta)  \frac{ \Delta }{  \ln  \Delta}  + (1 + \fdelta) ( 1 + \fdelta/2) \frac{ \Delta} { \ln \sqrt{  f } } \\
            & \le &  ( 1 + 2\fdelta)  \frac{ \Delta }{  \ln \sqrt{ f }}  + (1 + \fdelta) ( 1 + \fdelta/2) \frac{ \Delta} { \ln \sqrt{  f } } \\
            & = & \left(2+\frac{7\fdelta}{2} +\frac{\fdelta^2}{2}\right) \frac{ \Delta }{ \ln \sqrt{f} } \\
            & \le & (2 + \epsilon) \frac{ \Delta }{ \ln \sqrt{f} } .
\end{eqnarray*} 

\end{proof}

Theorem~\ref{sparse_case} follows immediately from  the following lemma, whose proof is similar to  Lemma 2.3 in~\cite{alon1999coloring} and can be found in Section~\ref{proof_subgraphs_few_triangles}.  The proof of Lemma~\ref{subgraphs_few_triangles} uses the standard Local Lemma and Theorem~\ref{improvement_mike},  so it can be made constructive using the Moser-Tardos  algorithm and the algorithm in Theorem~\ref{improvement_mike}.

\begin{lemma}\label{subgraphs_few_triangles}
For every $\fepsilon , \fdelta  \in (0,1)$   there exists $\Delta_{\fepsilon, \fdelta}> 0 $  such that  for every graph $G = (V,E)$  with maximum degree $\Delta \ge \Delta_{\fepsilon, \fdelta}$ in which the neighbors of every vertex span at most $\Delta^{2 - (2+ \fdelta) \fepsilon}$ edges, there exists a partition of the vertex set $V = V_1 \cup \ldots \cup V_k$ with $k=  \Delta^{1- \fepsilon}    $, such that for every $1 \le i \le k$,
\begin{align*}
\chi(G[V_i]) \le  (1+\fdelta)(1 + \fepsilon^{-1})  \frac{ \Delta^{\fepsilon} }{  \ln \Delta }.
\end{align*}
\end{lemma}

\begin{proof}[Proof of Theorem~\ref{sparse_case}]
If $V_1, V_2, \ldots, V_k$, $k = \Delta^{1- \fepsilon}$ is the partition promised by 
Lemma~\ref{subgraphs_few_triangles} %\ref{sparse_case}, 
then
\begin{align*}
\chi(G) \le \sum_{i=1}^{ \Delta^{ 1- \fepsilon} } \chi(G[V_i])  \le ( 1 +  \fdelta ) (1+ \fepsilon^{-1} ) \frac{\Delta}{ \ln \Delta } .
\end{align*}
\end{proof}

To prove Theorem~\ref{sparse_graphs} for $f \in [ f_{\epsilon}, \Delta^{ (2 + \epsilon^2) \epsilon } )$, we will perform a sequence of random halving steps, as in~\cite{alon1999coloring}, to partition the graph into  subgraphs satisfying the condition of Theorem~\ref{sparse_graphs} 
%\red{with applied with parameter $\zeta = \epsilon^2$ for $f \ge \Delta^{ ( 2 + \zeta^2) \zeta } $ and color these subgraphs using disjoint sets of colors. }
with $f \ge \Delta^{ ( 2 + \epsilon^2) \epsilon } $ and color these subgraphs using disjoint sets of colors. 
To perform the partition we use the following lemma from~\cite{alon1999coloring}. As it is proven via the standard LLL, it can be made constructive using the Moser-Tardos algorithm.
\begin{lemma}[\cite{alon1999coloring}]\label{bisection}
Let $G(V,E)$ be a graph with maximum degree $\Delta \ge 2$ in which the neighbors of every vertex span at most $s$ edges. There exists a partition $V = V_1 \cup V_2$ such that
the induced subgraph $G[V_i], i = 1,2$, has maximum degree at most $\Delta/2 + 2 \sqrt{ \Delta \ln \Delta}$ and the neighborhors of every vertex in $G[V_i], i =1,2$, span at most $s/4 + 2 \Delta^{ \frac{3}{2} } \sqrt{  \ln  \Delta} $ edges.
\end{lemma}

We will also use the following lemma whose proof, presented in Section~\ref{param_lemma_proof}, is almost identical to a similar statement in the proof of Theorem 1.1 of~\cite{alon1999coloring}. 
\begin{lemma}\label{param_lemma}
Given $\Delta, f$ sufficiently large, let the sequences $\Delta_t$ and $s_t$ be defined as follows. $\Delta_0 = \Delta, s_0 = \Delta^2/f$ and 
\begin{align*}
\Delta_{t+1} = \Delta_t/2 + 2 \sqrt{ \Delta_t  \ln \Delta_t }, \enspace s_{t+1} = s_t/4 + 2 \Delta_t^{ \frac{3}{2} } \sqrt{ \ln \Delta_t } . 
\end{align*}
For any $\delta \in (0, 1/100)$ and $\zeta > 0$ such that $ \zeta (2 + \delta) < 1/10$, let $j$ be the smallest integer for which $f > \left( \frac{ \param \Delta }{ 2^j} \right)^{ (2+ \delta) \zeta} $. Then $\Delta_j \le  \param \Delta/2^j$ and $s_j \le \left(\param \Delta/2^j\right)^2/f$.
\end{lemma}

\begin{proof}[Proof of Theorem~\ref{sparse_graphs} for $f \in [ f_{\epsilon}, \Delta^{ (2 + \epsilon^2) \epsilon } )$]
Let $\epsilon_0 = 1/11$. For $\epsilon \in (0, \epsilon_0]$, let $\delta = \zeta = \epsilon^2$. Since $\zeta(2+\delta) < 1/10$, apply Lemma~\ref{param_lemma} and let $j = j(\Delta, f, \delta, \zeta)$ be the integer described therein. Let $S$ be the process which, given a graph $G$, does nothing if $\Delta(G) < 2$, and otherwise partitions $G$ as described in Lemma~\ref{bisection}. Apply $S$ to $G$ to get subgraphs  $G[V_1], G[V_2]$. Apply $S$ to $G[V_1],G[V_2]$ to get $G[V_{1,1}], G[V_{1,2}], G[V_{2,1}], G[V_{2,2}]$.  Repeating these steps $j$ times, we obtain a partition of $G$ into at most $2^j$ induced subgraphs. Observe that for each such subgraph $H$, either $\Delta (G) < 2$ and hence $\chi(H) \le 2$, or, by Lemma~\ref{param_lemma}, $\Delta(H) \le \param \Delta/2^j =: \Delta_{\ast}$ and the neighbors of every vertex in $H$ span at most $\Delta_{\ast}^2 / f$ edges, where $f \ge \Delta_{\ast}^{(2+\delta)\zeta} = 
\Delta_{\ast}^{(2+\zeta)\zeta} \ge \Delta_{\ast}^{(2+\zeta^2)\zeta}$. Therefore, by the already established case of Theorem~\ref{sparse_graphs},  either $\chi(H) \le 2$ or $\chi(H) \le (2+ \zeta) \Delta_{\ast}/\ln \sqrt{f}$. Thus,
\begin{align*} 
\chi(G) \le 2^j  \max \left\{  2 , (2+ \zeta) \frac{ \param \Delta / 2^j}{ \ln \sqrt{f} }  \right\}  
\le \max \left\{2^{j+1},  (2+\zeta) \frac{ (1+ \delta)  \Delta }{ \ln \sqrt{f}}  \right\}  
%\le  (2+\zeta) \frac{ (1+ \delta)  \Delta }{ \ln \sqrt{f}}  \enspace,
.
\end{align*}
To bound $2^{j+1}$ from above we first observe that, for all $f$ sufficiently large, i.e., for all $f \ge f_{\epsilon} $, 
\begin{align}\label{eq:fot_xest}
 \left( \frac{ \param \Delta }{ \frac{\Delta}{2 \ln \sqrt{f}  }} \right)^{ (2+ \delta) \zeta} = \left( 2( 1+ \delta) \ln \sqrt{f} \right)^{( 2 + \delta) \zeta} <  f  .
\end{align}
Now, since $j$ was defined as the smallest integer for which $ \left( \frac{ \param \Delta }{ 2^j} \right)^{ (2+ \delta) \zeta} <f $, we see that~\eqref{eq:fot_xest} implies $2^{j} \le  \frac{\Delta}{2 \ln \sqrt{f}  }$ and, therefore,   $2^{j+1} \le \frac{\Delta}{ \ln \sqrt{f} } $. Finally, we observe that $(2+\zeta) (1+\delta) = (2+\epsilon^2)(1+\epsilon^2) < 2+ \epsilon$ for all $\epsilon \in (0, \epsilon_0]$. Therefore, as claimed,
\[
\chi(G) \le (2+\epsilon) \frac{\Delta }{ \ln \sqrt{f}}  .
\]
\end{proof}

\subsection{Proof of Lemma~\ref{subgraphs_few_triangles}}\label{proof_subgraphs_few_triangles}

We follow an  approach similar to that of Lemma 2.3 in~\cite{alon1999coloring}, making appropriate modifications as needed. First we partition the  vertices of $G$ into $\Delta^{1- \fepsilon}$ parts by coloring them randomly and independently  with $\Delta^{ 1 - \fepsilon}$ colors. For a vertex $v$, call a neighbor~$u$ of~$v$ a  \emph{bad neighbor} if $u$ and $v$ have at least $\Delta^{1- (1+ \fdelta/2) \fepsilon}$ common neighbors. Otherwise, say that $u$ is a \emph{good neighbor}. Since the neighbors of every vertex span at most $\Delta^{ 2 - (2 + \fdelta) \fepsilon }$ edges, there are at most $2 \Delta^{1- (1+ \fdelta/2) \fepsilon}$ bad neighbors for any vertex in $G$. 

For any vertex $v$, define three types of bad event with respect to the random partitioning experiment: 
\begin{itemize}
\item
$A_v$: more than $(1+ \fepsilon) \Delta^{\fepsilon}$ neighbors of $v$ receive the same color as $v$. 
\item
$B_v$: more than $\frac{10 }{  \fepsilon \fdelta}$ bad neighbors of $v$ receive the same color as $v$. 
\item
$C_v$: the good neighbors of $v$ that receive the same color as $v$ span more than $\frac{100}{(\fepsilon \fdelta)^2 } $ edges. 
\end{itemize}
We will use the symmetric version of the Local Lemma~\cite{LLL} to show that we can find a coloring of the graph that avoids all bad events. First, note that each of the bad events $A_v, B_v, C_v $ is independent of all but at most $\Delta^2$ others, as it independent of all events $A_u, B_u, C_u$ corresponding to vertices $u$ whose distance from $v$ is more than $2$. Since the degree of  any vertex in its color class is binomially distributed with mean at most $\Delta^{\fepsilon}$,  standard Chernoff estimates imply that the probability that  $v$ has more than $(1+ \fepsilon) \Delta^{\fepsilon}$ neighbors with the same color as~$v$ is at most $\mathrm{e}^{- \Omega(\Delta^{\fepsilon}) }$, which means that $\Pr[ A_v ] <   \Delta^{-3}$ for large enough $\Delta$. Moreover, we also have
\begin{align*}
\Pr[B_v] \le  \binom{2\Delta^{1 - (1 + \fdelta/2)  \fepsilon}}{\frac{10}{\fepsilon \fdelta}} \left(  \frac{1}{ \Delta^{1-\fepsilon} }  \right)^{ \frac{10}{\fepsilon \fdelta} } \le  \left(  \frac{2}{  \Delta^{  \frac{ \fepsilon\fdelta}{2} } }  \right)^{\frac{10}{\fepsilon \fdelta}  } \le\Delta^{-3} ,
\end{align*}
for large enough $\Delta$. Finally, to bound the probability of $C_v$ we make the following observation. If a graph has at least $e^2$ edges, then either it has a vertex of degree at least $e$, or every vertex has degree strictly less than $e$, implying that the graph can be edge-colored with $\lceil e \rceil$ colors, in which case the largest color class must contain at least $e^2/\lceil e \rceil \ge e-1$ edges. Thus, a graph with more than $100/(\fepsilon \fdelta )^2$ edges either has a vertex of degree at least $10/ (\fepsilon \fdelta) \ge 9/(\fepsilon \fdelta)$ or a matching with at least $10/ (\fepsilon \fdelta) - 1 \ge 9 / ( \fepsilon  \fdelta)$ edges, where the inequality follows from the fact that $\fepsilon, \fdelta < 1$. Thus, $C_v$ can happen only if there is a good neighbor $u$ of $v$ such that $u$ and $v$ have at least   $9 / ( \fepsilon  \fdelta)$ common neighbors with the same color as $v$, or if there is a matching of size at least  $9 / ( \fepsilon  \fdelta)$ on the good neighbors of $v$ that have the same color as $v$. The probabilities of the first  and second  of these events are bounded, respectively, by
\begin{align*}
\Delta \binom{  {  \Delta^{1 - (1 + \fdelta/2) \fepsilon }}}{ \frac{9} {  \fepsilon  \fdelta}} \left(  \frac{1}{ \Delta^{1- \fepsilon} }   \right)^{  \frac{9}{ \fepsilon \fdelta}} \le  \left( \frac{1}{  \Delta^{ \frac{\fepsilon \fdelta}{2} } }    \right)^{ \frac{9}{\fepsilon \fdelta} }    \le \frac{1}{2} \Delta^{-3} ,
\end{align*}\begin{align*}
\binom{  \Delta^{2 - (2+ \fdelta) \fepsilon }}{ \frac{9} {  \fepsilon  \fdelta}} \left(  \left(  \frac{1}{ \Delta^{1- \fepsilon}  } \right)^{2}  \right)^{ \frac{9}{ \fepsilon \fdelta} } \le \left( \frac{1}{\Delta^{  \fepsilon \fdelta } }  \right)^{ \frac{9}{\fepsilon \fdelta} }   \le  \frac{1}{2} \Delta^{-3} .
\end{align*}
Therefore the probability of $C_v$ is at most $\Delta^{-3}$. Thus, the Local Lemma applies since each bad event has probability at most $\Delta^{-3}$ and is independent of all but at most $\Delta^2$ other bad events. This means that we can find a partition $V = V_1, \ldots, V_{k}$, where $k=\Delta^{1- \fepsilon}$, so that in each induced subgraph $G[V_i]$, every vertex
has degree at most $(1+ \fepsilon)\Delta^{ \fepsilon}$,  has at most $\frac{10}{\fepsilon\fdelta } $ bad neighbors, and is contained in at most $\frac{100} {  (\fepsilon \fdelta)^2 }$ triangles in which both other vertices are good. We will show that, given such a partition, each $G[V_i]$ can be colored with at most $\frac{ (1+\fdelta)(1 + \fepsilon^{-1})  \Delta^{\fepsilon} }{  \ln \Delta } $ colors, assuming $\Delta$ is large enough.

To see this, consider the partition $B_i, V_i \setminus B_i$ of $V_i$, where $B_i$ is the set of vertices $u \in V_i$ for which there exists a vertex $v \in V_i$, such that $u$ is a bad neighbor of $v$. We claim that $\chi(G[B_i] )  \le  \frac{20}{\fepsilon \fdelta } +1   $ and $\chi( G[ V_i \setminus B_i]  ) \le  \frac{ (1+\fdelta/2)(1 + \fepsilon^{-1})  \Delta^{\fepsilon} }{  \ln \Delta } $. Assuming this claim, observe that
\begin{align*}
\chi(G[V_i] ) \le \frac{20}{\fepsilon \fdelta} +1 +\frac{ (1+\fdelta/2)(1 + \fepsilon^{-1})  \Delta^{\fepsilon} }{  \ln \Delta } \le \frac{ (1+\fdelta)(1 + \fepsilon^{-1})  \Delta^{\fepsilon} }{  \ln \Delta }  ,
\end{align*}
where the last inequality holds for all $\Delta \ge \Delta_{\fepsilon, \fdelta}$.

To see the first part of the claim, note that it is well-known (and easy to see) that if a graph has an orientation with maximum outdegree $d$, then it is ($2d+1$)-colorable. Consider the orientation of the graph on $B_i$ that results when every vertex points to its bad neighbors in $B_i$. Clearly, the maximum outdegree is at most $\frac{10}{\fepsilon \fdelta}$ and, thus, $\chi(G[B_i] )  \le  \frac{20}{\fepsilon \fdelta } +1   $ .

To see the second part of the claim, observe that each vertex of $G[V_i \setminus B_i ] $ is contained in at most $\frac{100}{(\fepsilon \fdelta)^2 }$ triangles.  Let $\Delta_{\ast} = (1+\fepsilon)\Delta^{\fepsilon  }$ and
\begin{align*}
f =  \frac{ \left( (1+\fepsilon)\Delta^{\fepsilon  } \right)^2 }{ 100   / (\fepsilon \fdelta)^2}    
= \frac{  (\fepsilon \fdelta)^2 \Delta_{\ast}^2 }{ 100}   
\ge   \Delta_{\ast}^{ \frac{ 2+  \frac{2\fdelta}{3} }{ 1+ \frac{2\fdelta}{3} }  } \left( \ln   \Delta_{\ast} \right)^2 ,
\end{align*}
where the last inequality holds for  $\Delta \ge \Delta_{\fepsilon,\fdelta}$. Applying Theorem~\ref{improvement_mike} to $G[V_i \setminus B_i ]$ (plugging in $\fdelta/3$ for the $\epsilon$ in that theorem) we get that, for all $\Delta \ge \Delta_{\fepsilon, \fdelta}$, 
\begin{eqnarray*}\label{long_manip}
\chi_{\ell}(G[ V_i \setminus B_i] ) & \le & (1+ \fdelta/3)  \frac{ \Delta_{\ast} }{ \ln \sqrt{ f}  }  \\
& = &  (1 +\fdelta/3)  \frac{(1+ \fepsilon)  \Delta^{\fepsilon}}{   \ln  \frac{ (1+ \fepsilon) \fepsilon \fdelta \Delta^{ \fepsilon} }{  10 }  } \\
& = &  (1 +\fdelta/3)  \frac{(1+ \fepsilon)  \Delta^{\fepsilon}}{\fepsilon \ln \Delta +    \ln  \frac{ (1+ \fepsilon) \fepsilon \fdelta }{  10 }  }\\
& \le & (1 +  \fdelta/2) \frac{(1 + \fepsilon)}{ \fepsilon }  \frac{ \Delta^{ \fepsilon} }{ \ln \Delta }  \\
& = & (1 +  \fdelta/2 )  \frac{ (1 + \fepsilon^{-1})\Delta^{ \fepsilon} }{ \ln \Delta }   ,
\end{eqnarray*}
as claimed.

\subsection{Proof of Lemma~\ref{param_lemma}}\label{param_lemma_proof}

Let $\epsilon' := \zeta( 2 + \delta) $ and recall that $\epsilon' < \frac{1}{10}$ by hypothesis. By the definition of $j$, for every $ t < j$, 
 \begin{align*}
\Delta_t \ge \Delta / (2^t)  >  \frac{ f^{ \frac{1}{\epsilon'} }}{1+ \delta },
\end{align*}
and $f^{ \frac{1}{\epsilon'} }/(1+ \delta)$ can be made arbitrarily large by taking $f$ sufficiently large. Hence, we can assume that $\Delta_t$ is sufficiently large so that $\Delta_{t+1} \le \frac{\Delta_t }{2} +  \Delta_t^{  \frac{ 2}{3} } \le \frac{1}{2} \Bigl(   \Delta_t^{\frac{1}{3} } +1 \Bigr)^3$. Taking cube roots and subtracting  $\frac{1}{2^{ \frac{1}{3} }-1 }$ from both sides we get
\begin{align*}
\Delta_{t+1}^{  \frac{1}{3}} - \frac{1}{2^{\frac{1}{3}}-1} \le \frac{1}{2^{\frac{1}{3} }}(  \Delta_t^{ \frac{1}{3} } +1 ) - \frac{1}{2^{ \frac{1}{3} } -1 }   = \frac{1}{2^{\frac{1}{3} }} \left( \Delta_t^{ \frac{1}{3} } - \frac{1}{2^{ \frac{1}{3} } -1}   \right) .
\end{align*}
Therefore, 
\begin{align}\label{j_steps}
\Delta_j^{ \frac{1}{3} }  - \frac{1}{2^{ \frac{1}{3} }  -1 } \le \frac{1}{ 2^{j/3}} \left(  \Delta_0^{ \frac{1}{3} } - \frac{1}{ 2^{ \frac{1}{3} }-1 } \right) .
\end{align}
Since $\Delta_0 = \Delta$, $2^{ \frac{1}{3} } -1  > \frac{1}{4}$ and $\Delta / 2^{j-1}  > \frac{ f^{ \frac{1}{\epsilon'} }}{1+ \delta }$ is large enough,~\eqref{j_steps} implies that
\begin{align*}
\Delta_j^{ \frac{1}{3} } \le \frac{ \Delta^{ \frac{1}{3} }}{ 2^{j /3}} + 4 \le  (1 + \delta)^{ \frac{1}{3} } \frac{\Delta^{ \frac{1}{3} } }{ 2^{ j/3}} .
\end{align*}
Therefore, we have shown that  $\Delta_j  \le \param  \frac{ \Delta }{ 2^j}$. Note also that the same proof shows that for every $t \le j$ we have $ \Delta_t \le \param \frac{ \Delta}{2^t } $.

We turn now to  the claim regarding $s_j$. For all $t < j$, we have by definition
\begin{align}
s_t \ge \frac{s_0 }{ 4^t} = \frac{ \Delta^2}{ 4^t f }  = \frac{1}{ \param^2}  \frac{ \left(  \frac{ \param  \Delta}{2^t }   \right)^2 }{f }  \ge   \frac{ 1}{ \param^2   } \left(  \frac{ \param \Delta }{ 2^t }   \right)^{ 2- \epsilon' }   \ge \frac{ 1}{  \param^2}  \Delta_t^{2  - \epsilon' }  \label{bound_delta_t},
\end{align}
where in the last inequality we used the fact that $\Delta_t  \le \param  \frac{ \Delta }{ 2^t}$ for all $ t\le j$. Using~\eqref{bound_delta_t} to bound $\Delta_t$ in the expression that defines  $s_t$, we get
\begin{align}\label{eq:ouga}
s_{t+1}  \le \frac{s_t}{4} +  2( \param^2  s_t )^{ \frac{3}{2}  \frac{ 1}{(2- 2 \epsilon' )} }  . 
\end{align}
To bound the r.h.s. of~\eqref{eq:ouga} we recall that $\epsilon' < \frac{1}{10}$ implying $\frac{3}{2(2-2\epsilon')} < \frac{3}{2(2-1/5)} = 5/6$. Assuming that $f$ (and hence also $\Delta_t$) is large enough, we obtain
\begin{equation}\label{xi}
s_{t+1}  \le \frac{s_t}{4} +  3 s_t^{5/6} \le \frac{1}{4} (s_t^{1/6}+2)^6 
= \frac{1}{4}(s_t + 12 s_t^{5/6} + 60 s_t^{2/3} + \cdots)  .
\end{equation}
Taking $6$-th roots and subtracting $\frac{5}{6^{1/6} -1 }$ from both sides, we obtain
\begin{align*}
s_{t+1}^{ 1/6} - \frac{5}{ 6^{1/6 } -1}  
\le \frac{1}{4^{1/6} } ( s_t^{ 1/6} +2 ) - \frac{5}{ 6^{1/6 }-1} 
\le \frac{1}{ 4^{ 1/6} } \left(  s_t^{ 1/6} - \frac{5}{ 6^{ 1/6}-1 }   \right) .
\end{align*}
Therefore,
\begin{align*}
s_j^{1/6} - \frac{5 }{6^{ 1/6} -1  } \le \frac{1}{ 4^{ j/6 } }  \left( s_0^{  1/6} - \frac{5}{6^{1/6} -1 }   \right) ,
\end{align*}
and, since $s_0 = \Delta^2/f$, 
\begin{align*}
s_j^{ 1/6} \le  
\frac{ (\Delta^{ 2})^{1/6} }{ 4^{ j/6 } f^{ 1/6}  }  + \frac{5}{ 6^{1/6} -1 }  
\le
\left(\frac{\Delta^{ 2} }{ 4^{ j} f  }\right)^{1/6}  + 15  .
\end{align*}

Since 
\[
\frac{\Delta^2}{4^j f}  = \left(\frac{\Delta}{2^j} \right)^2  \frac{1}{f} 
= \left(\frac{\Delta}{2^{j-1}}\right)^2  \frac{1}{4f} > \left(\frac{f^{1/\epsilon'}}{1+\delta}\right)^2 \frac{1}{4f}
= \frac{ f^{ \frac{2}{\epsilon'}-1 }}{4(1+ \delta)^2 }
\]
can be made arbitrarily large by taking $f$ sufficiently large, we see that 
\[
\left(\frac{\Delta^{ 2} }{ 4^{ j} f  }\right)^{1/6}  + 15 
\le ( 1+ \delta)^{ 1/3}  \left(\frac{\Delta^{ 2} }{ 4^{ j} f  }\right)^{1/6}  .
\]
Thus $s_j \le (\param \Delta / 2^j)^2/f  $, completing the proof.

\section{Proof of Proposition~\ref{Random_Graphs}}\label{random_graphs_proof}

We use the term ``with high probability" to refer to probabilities that tend to $1$ as $n$ goes to infinity. Proposition~\ref{Random_Graphs} follows in a straightforward way from the following lemma. 

\begin{lemma}\label{random_lemma}
For any $\delta  \in (0,1)$  there exists a constant $d_{0}$ such that, for any $d  \in \left(  d_{0} \ln n,  ( n \ln n )^{\frac{1}{3} }  \right) $, each vertex of the random graph $G = G(n, d/n)$ is contained in at most $\Delta^{ \delta }$ triangles with high probability, where $\Delta$ is the maximum degree of $G$.
\end{lemma}

\begin{proof}[Proof of Proposition~\ref{Random_Graphs}]
According to~\cite{AchlioNaor}, for a graph $G \in G(n, d/n)$ we know that with high probability
\begin{align}
\chi(G) = \frac{1}{2} \frac{d}{\ln d} ( 1 + o(1) ) \label{chrom_random} .
\end{align}
 Fix $\zeta  \in (0,1) $ and  $\delta \in (0,  \frac{ 2 \zeta}{1 + 2 \zeta } )$. According to Lemma~\ref{random_lemma}, there exists a constant $d_0$ such that for any $d  \in \left(  d_{0} \ln n,  ( n \ln n )^{\frac{1}{3} }  \right) $ each vertex of   $G = G(n, d/n) $ is contained in at most $ \Delta^{ \delta }$ triangles with probability that tends to $1$ as $n$ goes to infinity. Thus, we can apply  Theorem~\ref{improvement_mike} with parameter $\zeta > 0$ since 
\begin{align*}
f = \frac{ \Delta^2}{ \Delta^{\delta}  } > \Delta^{ 2 - \frac{ 2 \zeta}{ 1 + 2 \zeta} }  (\ln \Delta)^2 , 
\end{align*} 
for large enough $\Delta$. This yields an upper bound $q$ on the chromatic number of $G$ that is at most
\begin{eqnarray}
q  & = & (1 + \zeta) \frac{  \Delta}{ \ln \sqrt{f} } \nonumber \\
    & \le & (1 + \zeta) \frac{ \Delta}{  \frac{ 1 + \zeta}{1 + 2 \zeta } \ln  \Delta + \ln \ln \Delta }  \nonumber \\
    & \le & (1 + 2\zeta)  \frac{  \Delta }{ \ln \Delta  } \label{random_upper_bound} .
\end{eqnarray}
Moreover, since the expected degree of every vertex of $G$ is $d$ and its distribution is binomial with parameter $\frac{d}{n}$, standard Chernoff bounds  and the union bound imply that for any $\eta \in (0,1)$, $ \Delta \le ( 1 + \eta ) d$ with high probability, for large enough $d_0$. 

Combining the latter fact with~\eqref{chrom_random} and~\eqref{random_upper_bound}, we deduce that we can find an arbitrarily small constant  $\eta' \in (0,1) $ such that
 \begin{align*}
 q \le (2+ \eta') \chi(G)
 \end{align*}
 by choosing $\zeta$ and $\eta$ sufficiently small.  Picking $\eta' = \frac{4 \epsilon  }{1 - 2 \epsilon } $ we obtain $ \chi(G)  \ge \frac{ q}{ 2 + \eta'} \ge q ( \frac{1}{2} - \epsilon) $, concluding the proof of Proposition~\ref{Random_Graphs}.
\end{proof}

Finally, we go back and prove Lemma~\ref{random_lemma}.

\begin{proof}[Proof of Lemma~\ref{random_lemma}]
Let $\Delta_v$ be the random variable that equals the degree of vertex $v$ of $G$. Observe that $\Delta_v \sim \mathrm{Binom}(n-1, \frac{d}{n} )$ and, therefore, using a standard Chernoff bound and the fact that $ d \ge d_0 \log n $ we get that
\begin{align*}
\Pr\left[ \Delta_v \notin  ( 1 \pm  \frac{1}{10})d \right]  \le \frac{1}{n^2} ,
\end{align*}
for large enough $d_0$. Thus, by a union bound we get that $\Pr[ \Delta \in  ( 1\pm \frac{1}{10}) d  ] \le \frac{1}{n}$.
Let $T_v$ be the number of triangles that contain vertex $v$ and  $B$ be the event that there exists a vertex $u$ such that  $\Delta_u \notin ( 1\pm \frac{1}{10}) d $. Then,
\begin{align}\label{bbb}
\Pr[T_v > \Delta^{\delta} ] \le  \Pr[ T_v > \Delta^{ \delta} \mid \overline{B} ] + \Pr [B ]  \le   \Pr[ T_v > \Delta^{ \delta} \mid \overline{B} ]  + \frac{1}{n}.
\end{align}
To upper bound the first term in the righthand side of~\eqref{bbb}, we fix an arbitrary  $\Delta_0 \in ( 1\pm \frac{1}{10}) d $, and upper bound $\Pr[ T_v > \Delta^{ \delta} \mid \Delta_v = \Delta_0 ]$. Towards that end, we observe that
\begin{align*}
\mathbb{E}[ T_v  \mid \Delta_v = \Delta_0 ]  = { \Delta_v \choose 2 } \frac{d}{n} \le \frac{( 1 + \frac{1}{10})^2  d^3}{2 n}.
\end{align*}
%Observe that $T_v \sim \mathrm{Binom} \left( \binom{n-1 }{ 2},  \left(  \frac{d }{n } \right)^3 \right)  $ and $\ex[ T_v ] \le   \frac{ d^3 } { 2n }   $. 

Thus, for any fixed values  of $\Delta, \Delta_0 \in (1 \pm  \frac{1}{10}) d$,  setting $1 + \beta =   \frac{ \Delta^{ \delta} }{( 1 + \frac{1}{10})^2 d^3 / 2n  }   $ and using a standard Chernoff bound we obtain:
\begin{align*}
\Pr [ T_v > \Delta^{ \delta} \mid \Delta_v =  \Delta_0  ] \le \mathrm{e}^{-  \frac{  \beta^2 1.1^2  d^3/2n} { 3}  }   \le \frac{1}{n^2}
\end{align*}
since
\begin{eqnarray*}
\beta &\ge & \frac{ \left((1- \frac{1}{10}) d \right)^{\delta}   - 1.1^2 d^3/2n }{ 1.1^2 d^3 /2n  }  > 0  , \\
\frac{1.1^2}{3} \beta^2 \frac{ d^3} {2n} &\ge & \frac{1.1^2}{3}  \frac{   \left( \left((1- \frac{1}{10}) d \right)^{\delta} - d^3/2n  \right)^2 }{ d^3/2n}   \ge 2 \ln n ,
\end{eqnarray*}
whenever $d \in [ d_0 \ln n, (n \ln n )^{ \frac{1}{3} } ]$ and for large enough $n$ and $d_0$. Since $\Delta_0$ was chosen arbitrarily, we obtain that $\Pr[ T_v > \Delta^{ \delta} \mid \overline{B} ] \le \frac{1}{n^2}$. Taking a union bound
over~$v $ concludes the proof of the lemma.
\end{proof}

\end{document}